\documentclass[lettersize,journal]{IEEEtran}

\pdfoutput=1 
\usepackage{hyperref}
\usepackage{graphicx}
\usepackage[all]{hypcap} 
\usepackage{color,overpic,cite}
\usepackage{xcolor}
\usepackage{pdfsync}
\usepackage{url}
\usepackage{amsmath,amsfonts,amssymb,amsthm,bbm,bm,comment,multirow, subfigure}
\usepackage{tcolorbox}
\usepackage{siunitx}
\usepackage{mdframed}
\usepackage{bm}
\usepackage[ruled,commentsnumbered]{algorithm2e}
\usepackage{tabularx,booktabs}
\usepackage[normalem]{ulem}
\usepackage{setspace}

\newtheorem{lemma}{Lemma}

\def\R{\mathcal{R}}
\def\X{\mathcal{X}}
\def\bE{\mathbb E}
\def\Y{\mathcal{Y}}

\newcommand{\revmajor}[1]{{\color{black}#1}} 
\newcommand{\revminor}[1]{{\color{black}#1}}

\begin{document}

\title{Adaptive Resource Allocation for Virtualized Base Stations in O-RAN with Online Learning}

\author{Michail~Kalntis,~\IEEEmembership{Graduate Student Member,~IEEE,}
        George~Iosifidis,~\IEEEmembership{Member,~IEEE,}
        and~Fernando~A.~Kuipers,~\IEEEmembership{Senior Member,~IEEE}
\thanks{The authors are with Delft University of Technology (emails: \{m.kalntis, g.iosifidis, f.a.kuipers\}@tudelft.nl).}}


\markboth{IEEE Transactions on Communications}%
{Kalntis,
\MakeLowercase{\textit{(et al.)}}:
Adaptive Resource Allocation for Virtualized Base Stations in O-RAN with Online Learning}

\maketitle

\begin{abstract}

Open RAN systems, with their virtualized base stations (vBSs), offer increased flexibility and reduced costs, vendor diversity, and interoperability. \revmajor{However, optimizing the allocation of radio resources in such systems raises new challenges due to the volatile vBSs operation, and the dynamic network conditions and user demands they are called to support. Leveraging the novel O-RAN multi-tier control architecture, we propose a new set of resource allocation \emph{threshold policies} with the aim of balancing the vBSs' performance and energy consumption in a robust and provably optimal fashion. To that end, we introduce an online learning algorithm that operates under minimal assumptions and without requiring knowledge of the environment, hence being suitable even for ``challenging'' environments with non-stationary or adversarial demands and conditions. We also develop a meta-learning scheme that utilizes other available algorithmic schemes, e.g., tailored for more ``easy'' environments, by choosing dynamically the best-performing algorithm; thus enhancing the system's effectiveness. We prove that the proposed solutions achieve sub-linear regret (zero optimality gap), and characterize their dependence on the main system parameters. The performance of the algorithms is evaluated with real-world data from a testbed, in stationary and adversarial conditions, indicating energy savings of up to \SI{64.5}{\percent} compared with several state-of-the-art benchmarks.}

\end{abstract}


\IEEEpeerreviewmaketitle

\section{Introduction}\label{sec:intro}

\subsection{Motivation \& Background}

\IEEEPARstart{T}{he} importance of base station virtualization is best illustrated by the flurry of industrial and academic activities that focus on the development of virtualized and Open Radio Access Network (O-RAN) architectures \cite{o-ran-andres}. The O-RAN Alliance, for example, is a global initiative aiming to softwarize and standardize RANs so as to improve their performance, reduce their costs, and lower the entry barrier towards a wider vendor ecosystem. At the core of this transformation are the virtualized Base Stations (vBSs), such as srsLTE \cite{gomez2016srslte} and OpenAirInterface (OAI) \cite{oai}, which offer OPEX/CAPEX savings and performance gains, since their operational parameters can be adjusted with high granularity at runtime \textcolor{black}{\cite{Alnoman_survey18}}. Alas, these benefits come at a cost. Softwarized base stations are found to have less predictable performance and more volatile energy consumption \cite{rost-globecom15, jose-icc21, ayala2021bayesian}, an {issue} that is amplified when instantiating them in general-purpose computing infrastructure. This induces operation and cost uncertainties at times when there is an increased need for \textcolor{black}{robustness} and performance guarantees in mobile networks. Therefore, it becomes imperative to understand how to configure or schedule these vBSs \textcolor{black}{(i.e., how to allocate their resources) without relying on strong assumptions or compromising network performance, in order to unblock their deployment and maintain energy costs at sustainable levels.}

\revmajor{The O-RAN architecture offers new opportunities to achieve this goal. Namely, the emerging O-RAN standards \cite{oran-spec-arch, oran-spec-scenarios} have provisions for multi-tier control solutions for resource management that can be implemented centrally, i.e., by the RAN Intelligent Controller (RIC), and enforced at different time-scales. In particular, our focus here is on non-Real-Time (non-RT) policies that determine the operation envelope (or resource allocation \emph{thresholds}) of the vBSs over time intervals (rounds) of a few seconds. These policies are fed to, and enforced by, the real-time radio scheduler of each vBS, which devises their assignments subject to global rules about, e.g., the maximum transmission power, the eligible modulation and coding schemes (MCS), and so on. Such centralized threshold policies have been recently introduced, e.g., see \cite{vrain_conf, ayala2021bayesian, edgebol21}, and have several practical advantages. First, O-RAN includes heterogeneous base stations that are challenging, if not impossible, to configure directly by intervening with their real-time schedulers. The global non-RT policies, on the other hand, offer an easy path to shape the operation of each vBS.  Secondly, using such central policies, the O-RAN controllers can coordinate the operation of their vBSs in a unified fashion, managing jointly their resources, and also use AI/ML mechanisms that can benefit from this centralized view.}

\revmajor{Nevertheless, the effective design of such policies is a new and particularly intricate problem. Due to their coarse time scale (seconds) and unlike the typical Radio Resource Management (RRM) decisions (updated in msecs), these policies do not have access to the network conditions and user traffic that will be realized during the interval they will be applied. And, further, these parameters can change arbitrarily during such large time windows, not necessarily following a stationary distribution. Moreover, due to the heterogeneity and volatile operation of the vBSs, the effect of such policies on the KPIs of interest is challenging, if not impossible, to predict or quantify with analytical expressions. Coupled with the typically large number of possible policies, this compounds finding the optimal policy for each vBS. In light of these observations, it is not surprising that the first works in this area focused on O-RAN operations under static network conditions and demands, \cite{ayala2021bayesian, edgebol21}.}

\revmajor{Our work addresses the following question: \emph{how to design robust vBS non-RT policies that offer performance/cost guarantees {without relying on strong assumptions} and avoid sub-optimal operation points?}  We consider O-RAN policies that determine thresholds (upper bounds) for key vBS operation knobs, namely for the vBS transmission power, the eligible MCS, and the Physical Resource Blocks (PRB), in the Uplink (UL) and Downlink (DL). Each policy is updated at a non-RT scale, based on the performance, cost, and context (conditions and demands) observations of the past, and is subsequently fed to real-time schedulers that assign the vBS radio resources.  }

\subsection{Contributions}
Our \emph{first contribution} is the design and evaluation of a robust \textcolor{black}{\emph{adversarial} bandit algorithm, cf. \cite{kalntis22}}, which: \emph{(i)} identifies effective policies without relying on assumptions \textcolor{black}{about the environment}; \emph{(ii)} offers \textcolor{black}{tight} performance guarantees; \emph{(iii)} is oblivious to the (unknown and possibly time-varying) vBS performance; and \emph{(iv)} has minimal and constant (in observations and time) memory requirements, as it uses closed-form expressions that can be calculated even in real-time and in resource-constrained platforms. The performance is quantified using a combined metric of effective throughput modulated by the \textcolor{black}{traffic} demands, and energy consumption, where the latter can be prioritized via a weight parameter. 
\revminor{It is important to note that no assumption (e.g., convexity) is made on the performance function (i.e., we follow a black-box approach).}
For the optimality criterion, we use \emph{regret}, where we compare the time-aggregated performance of the algorithm with that of a hypothetical benchmark that is designed with the help of an oracle providing access to all future/necessary information.

The \emph{second contribution} is the expansion of this learning algorithm with a \emph{meta-learning} scheme, which boosts the performance whenever possible. Namely, the robustness of the algorithm described above means it might be conservative when \textcolor{black}{the environment is \textit{easy}, e.g.,} when the network has access to context information, or if the \textcolor{black}{channel qualities} and \textcolor{black}{traffic} demands are stationary or exhibit periodicity \cite{marquez2020identifying}. For these cases, data-efficient solutions such as \cite{ayala2021bayesian} can leverage the available information to identify optimal policies faster. Hence, the question that arises naturally is how to combine the required robustness without compromising learning performance (in terms of convergence speed) whenever the \textcolor{black}{environment is easy}. To address this, we introduce a \emph{meta-learner} that selects intelligently among policies proposed by different algorithms that rely on, and perform better under, different assumptions. A key challenge is that the learning happens on two levels: the meta-learner has to learn which is the best-performing algorithm, and each algorithm has to learn which is the best-performing policy, while partial (i.e., bandit) feedback is received on both levels. Our approach addresses this challenge through a framework that guarantees the network will perform as well as the best-performing algorithm.

In summary, the main technical contributions of this paper are the following: 

\noindent $\bullet$ We study the vBS \textcolor{black}{resource allocation} problem in its most general form, i.e., in non-stationary/adversarial \textcolor{black}{environment} and without knowledge of vBS throughput/cost functions. Our proposed scheme achieves sub-linear regret and has minimal computation and memory overhead \textcolor{black}{\cite{kalntis22}}. This is the first work applying \emph{adversarial} bandits to vBS \textcolor{black}{resource allocation}.

\noindent $\bullet$ We devise a meta-learning strategy that entails the use of algorithms tailored to different environments and obtains sub-linear regret with respect to the best \textcolor{black}{algorithm}, in each case.

\noindent $\bullet$ We use real-world traffic traces and testbed measurements to demonstrate the weaknesses of prior works \cite{ayala2021bayesian}, as well as the efficacy of the proposed learning algorithm in a battery of representative scenarios. Upon publication of this article, we will release all the source code to foster further research on this important topic.

\subsection{Organization}

Sec. \ref{sec:related} discusses the related work, and Sec. \ref{sec:system_model} introduces the O-RAN background and the system model. In Sec. \ref{sec:configuration_learning_for_adversarial_environments}, we present the main learning algorithm, and in Sec. \ref{sec:universal_configuration_learning_through_an_experts_model}, the meta-learner. Sec. \ref{sec:evaluation} illustrates the performance evaluation and we conclude in Sec. \ref{sec:conclusions}. 
\section{Related Work}\label{sec:related}

\begin{figure*}[!t]
	\centering
	\subfigure[]{\label{fig:architecture_leftleft}\includegraphics[scale=0.36]{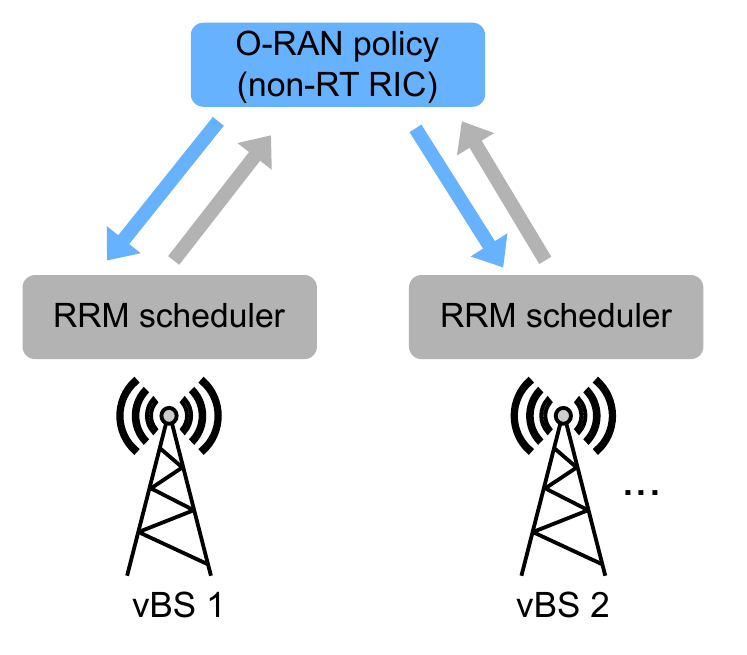}}	
	\subfigure[]{\label{fig:architecture_left}\includegraphics[scale=0.21]{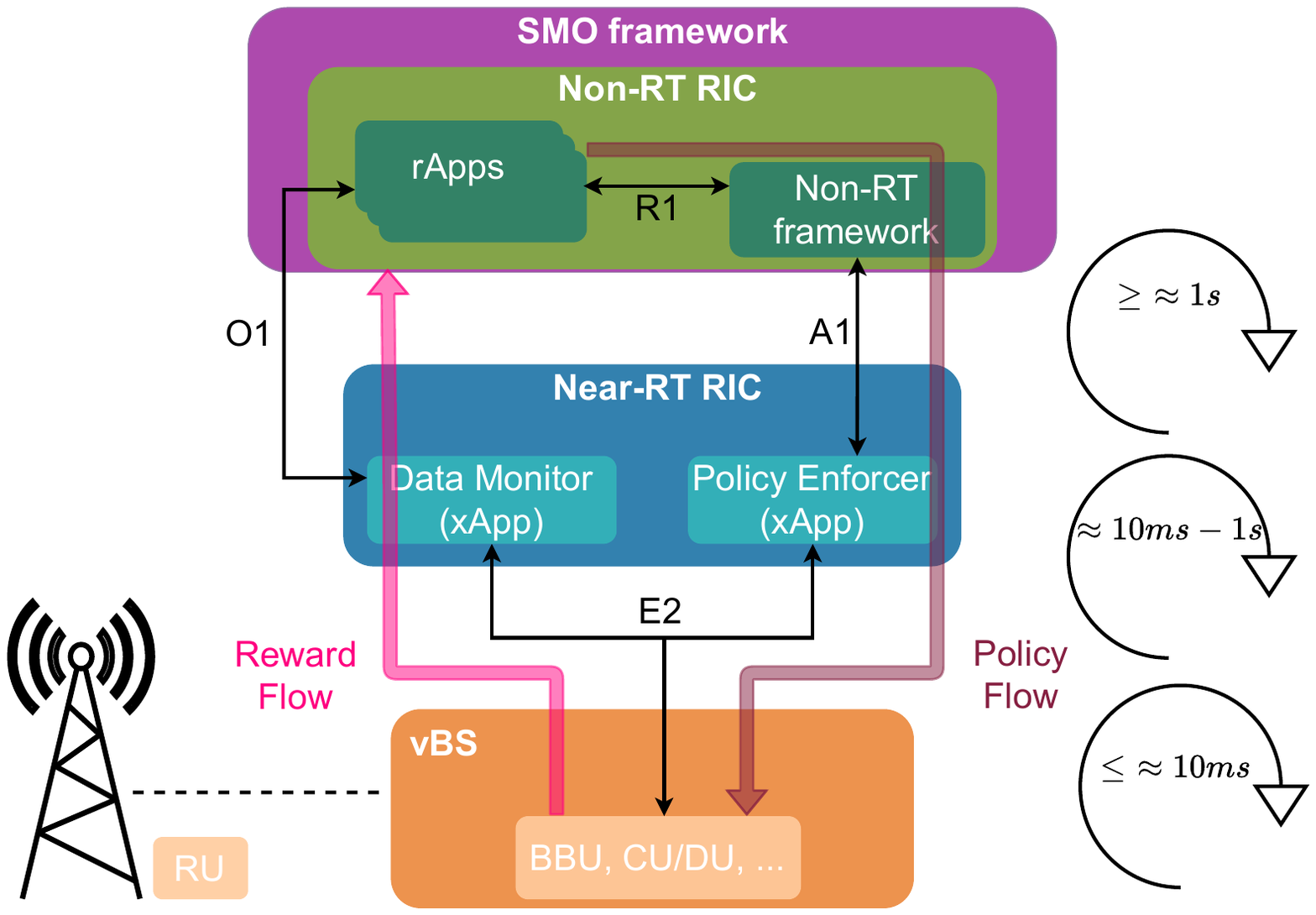}}
	\subfigure[]{\label{fig:architecture_right}\includegraphics[scale=0.21]{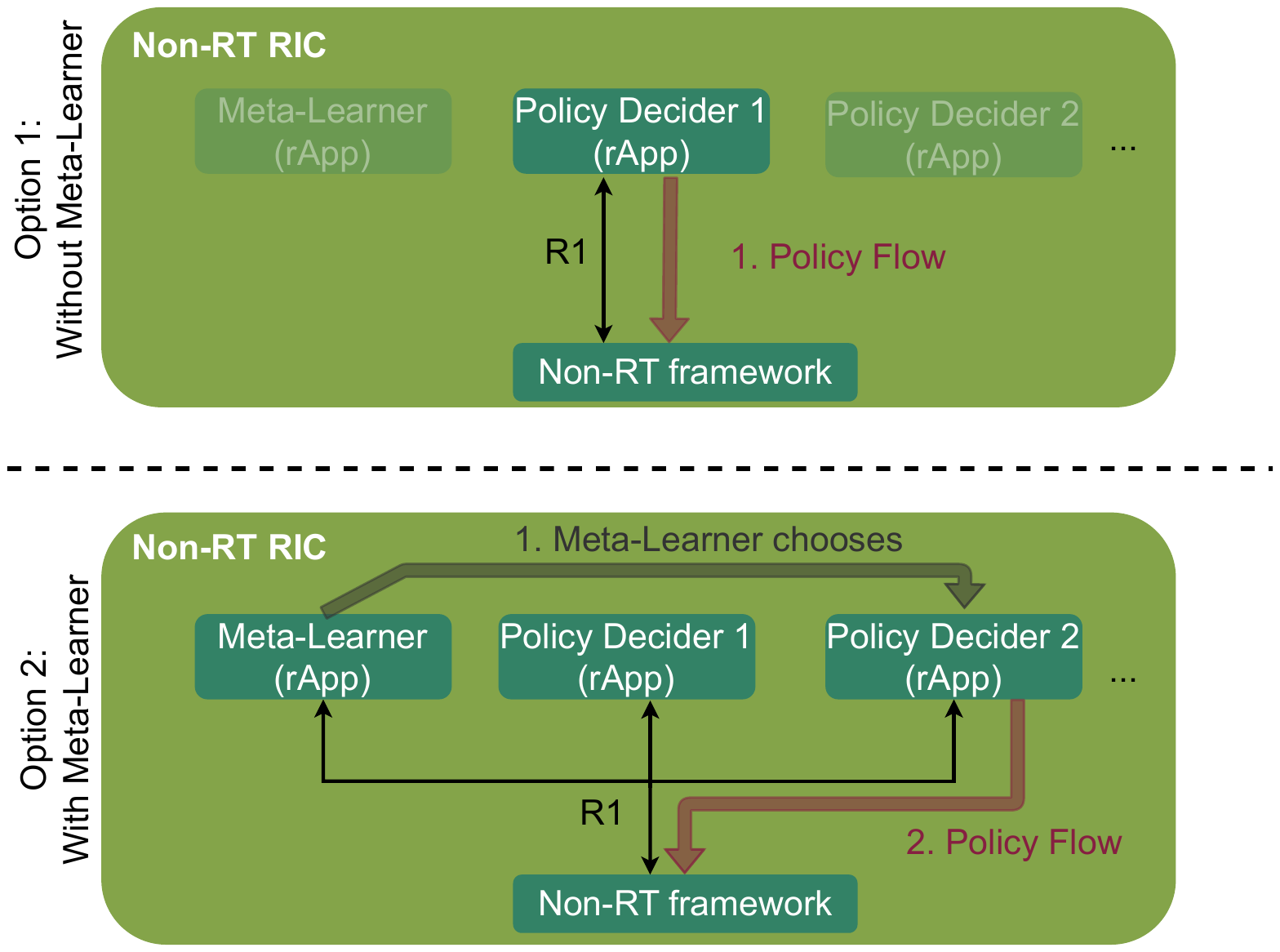}}
	\caption{\revmajor{O-RAN-compliant  architecture \& policy workflow.  \textbf{(a)}. The proposed policy operates in the non-RT RIC and decides MCS, Power and PRB thresholds that are sent to each vBS's scheduler.} \textbf{(b)} The key building block is the Non-RT RIC, hosted by the Service Management and Orchestration (SMO) framework, and the Near-RT RIC. The system has three control loops: (\emph{i}) Non-RT, which involves large-timescale operations with execution time \SI{\geq 1}{\second}, (\emph{ii}) Near-RT (\SI{>10}{\milli\second}), and (\emph{iii}) RT (\SI{\leq 10}{\milli\second}). \textcolor{black}{\textbf{(b)} Policy Flow for the Non-RT RIC with (bottom) and without (top) an rApp implementing a meta-learner.}}
	\label{fig:architecture}
\end{figure*}

\noindent \revmajor{
Resource management for \emph{softwarized} networks can be broadly classified into models that relate {policies} to performance functions, model-free approaches, and Reinforcement Learning (RL) techniques. Model-based examples include \cite{rost-globecom15} and \cite{bega2018cares}, which maximize the served traffic subject to vBS computing capacity, but do not capture the impact that the hosting platform, the environment, or user demands may have on the vBS's operation \cite{jose-icc21}. Model-free approaches employ Neural Networks to approximate the performance functions of interest \cite{patras-DNN-Tutorial2019}, yet, their efficacy depends on the availability of representative training data. Finally, RL solutions \cite{DL_wireless19} use runtime observations and have been used, for example, in interference management \cite{alcaraz2020online}, vBS function splitting \cite{murti21}, and handover optimization \cite{handover_cao22}.
The disadvantages of all these works are the curse of dimensionality and the lack of robust convergence guarantees \cite{RL_curse}. Following an akin approach, contextual bandit algorithms are employed to optimize video streaming rates \cite{bastian-cba-infocom19} or handover decisions \cite{chuai-collab-INFOCOM19}; assign CPU time to virtualized BSs  \cite{vrain_conf}, and control millimeter-Wave networks \cite{tekin-TCCN2020}. Unfortunately, these works require access to \emph{context} information. More recently, Bayesian learning has been used for RRM, see \cite{lorenzo_2021} and references therein, but these solutions also need access to context information and converge only under stationary conditions. 

We take here a different approach, based on \emph{adversarial} bandits, cf. \cite{bubeck_bandits}, which is robust to adversarial or non-stationary contexts (channel qualities and traffic demands), and has low memory and computation requirements. This latter feature is in stark contrast to RL (with sizeable memory space required to store all space-actions combinations) and Bayesian techniques \cite{ayala2021bayesian, lorenzo_2021} which involve slow matrix inversions \cite{freitas-tutorial-Proc2016}. Such adversarial/non-stationary environments are increasingly common due to highly volatile network conditions \cite{abdul_nonstat_cqi} and traffic demands \cite{paschos-comm-mag}. Furthermore, we draw ideas from the expert-learning paradigm \cite{expert-littlestone} and enrich our policy decisions with a meta-learning scheme that combines our adversarial learning algorithm (that can be at times conservative) with any other algorithm (e.g., \cite{ayala2021bayesian}) that performs better on more \emph{easy} scenarios where the environment is predictable (e.g., when stationary). This meta-algorithm obtains the best of both approaches, and succeeds in being both fast-learning and robust; an idea that has been used in online learning \cite{NIPS2014_01f78be6}, but not in network management.

We employ the above method to tackle a joint performance and energy cost optimization problem. Similar (in scope) formulations have been extensively studied in the literature. For instance, \cite{multi-obj-2}  considered a joint user association, spectrum, and power allocation model for throughput optimization; \cite{multi-obj-1} focused on spectrum and energy efficient beamforming; and \cite{multi-obj-3} optimized the spectrum and power assignment using genetic algorithms. Nevertheless, such approaches assume the system and user-related parameters to be fixed and known. On the other hand, many dynamic formulations rely on  RL to optimize energy and performance, e.g., \cite{deep-learn-twc19}; or on variants of the seminal CGP-UCB algorithm \cite{ayala2021bayesian}. The main limitation of these works is the need to know contextual information (channels, user demands, etc.), and the lack of optimality guarantees, mainly in non-stationary conditions. Our approach is instead tailored to handle the inherent performance and cost volatility of O-RAN systems without access to context and provides optimality guarantees against competitive oracles.

\revminor{Finally, a key difference between our work and the above RRM literature lies in our emphasis on non-RT RAN policies. These policies serve as operational thresholds for the real-time vBS (i.e., RRM) schedulers and are facilitated by the O-RAN architecture, which has provisions for such tiered control loops \cite{oran-spec-arch, oran-spec-scenarios}. This approach enables centralized management of multiple BSs without disrupting their RRM functionality. 
Recent works \cite{tsampazi_oran, coloran} use RL for selecting the slicing and scheduling policies in O-RAN (i.e., RRM schedulers, see Sec. \ref{sec:system_model}). Nevertheless, our policy thresholds operate on a higher timescale (i.e., Non-RT) and are fed to these vBS schedulers, which make RRM decisions subject to our provided thresholds; also, they learn in an online (not offline) manner and adapt to environment changes, even if these changes happen drastically.}

A recent stream of works has followed this path to design such non-RT operation thresholds. Namely, \cite{ayala2021bayesian} uses CGP-UCB to identify thresholds for the maximum allowed transmission power and MCS to reduce the energy consumption of base stations; \cite{edgebol21} follows a similar approach but focuses on different performance KPIs; and \cite{vrain_conf} decides maximum MCS and duty cycles through deep learning. Unlike these works, our solution is the first to provide optimality guarantees for non-stationary environments and without requiring access to context, while through the proposed meta-learner we can combine and benefit from other algorithms (e.g., \cite{ayala2021bayesian}) when they perform well.
}

\section{Background \& System Model}\label{sec:system_model}

\subsection{O-RAN Background}

\revmajor{\revminor{Our model follows the O-RAN architecture \cite{oran-spec-arch, oran-spec-scenarios, o-ran-andres}, which has provisions for resource management and decision-making at different time scales: at seconds/minutes level (i.e., in Non-RT RIC through rApps\footnote{\revminor{The O-RAN general terms rApp and xApp are used to describe applications in the RIC at Non-RT (for rApp) or near-RT (for xApp) time scale, which can implement various RAN control algorithms or other pertinent services; see \cite{o-ran-andres}. }}), and the millisecond level (i.e., in Near-RT RIC through xApps). The proposed algorithms can be implemented as rApps at the Non-RT RIC, aiming to learn energy-efficient threshold radio {policies} \cite{oran-spec-scenarios}. These policies are essentially \emph{threshold rules} regarding the maximum MCS, PRB, and transmission power that each vBS, in real-time, is allowed to use. Specifically, these rules are communicated to the vBSs under the RIC, so as to guide their RRM schedulers which allocate the radio resources in real-time accordingly, see Fig. \ref{fig:architecture_leftleft}. This approach is in line with a recent stream of papers \cite{edgebol21, ayala2021bayesian, o-ran-andres, vrain_conf} proposing threshold policies and exploits the multi-tier (multi-timescale) architecture of O-RAN to offer centralized control of multiple vBS, without intervening to their (often proprietary/hardcoded) real-time schedulers.}

We consider here the typical vBS comprising a Base Band Unit (BBU) hosted by an off-the-shelf platform attached to a Radio Unit (RU).\footnote{The BBU corresponds to a Long-Term Evolution (LTE) eNodeB (eNB) for a 4G network and to a New Radio (NR) gNodeB (gNB) for a 5G network. For the latter, gNB is disaggregated into the Distributed Unit (DU), and the Centralized Unit (CU).} This tiered control approach can be seen in Policy Flow, Fig. \ref{fig:architecture_left} and \ref{fig:architecture_right} top. At each round, with typical duration $\!\sim\! 1$ second, the \emph{Policy Decider} (i.e., algorithm) devises the threshold policy which is communicated (via the A1 interface) to the Near-Real-Time (Near-RT) RIC, where an xApp (termed \emph{Policy Enforcer}) forwards it to the different vBSs\footnote{E2 nodes refer to O-RAN nodes O-CU, O-DU, O-RU, and O-eNB, which denote the CU, DU, RU, and eNB, respectively.}. This makes a two timescale system where the policy is devised at each round (seconds) and the vBSs schedulers update their typical RRM decisions every slot (mseconds), based on these rules.} 

\revminor{O-RAN's flexibility enables the usage of O1 to receive/forward the policy directly from/to the real-time scheduler \cite{coloran}. Nevertheless, our decision to involve xApps through the Near-RT RIC stems from providing a more general framework, where, e.g., another xApp could take the thresholds we provide, save them to a database, and perform additional actions to ensure that those thresholds are respected or make any other inference. Our modular architecture is designed to be adaptable and general enough to accommodate this, and is in accordance with recent works \cite{ayala2021bayesian, o-ran-andres}.}

\revmajor{The Policy Flow changes when including a meta-learner as another rApp (Fig. \ref{fig:architecture_right} bottom), whose goal is to discern the best Policy Decider among the employed ones. This is achieved by selecting at each round one of the available Policy Deciders, which, in turn, chooses the threshold policy. At the end of each round, the Near-RT RIC's Data Monitor computes a \emph{reward} by aggregating the performance and cost measurements (for all slots) received via the E2 and feeds them to the selected Policy Decider via the O1 interface ({Reward Flow} in Fig. \ref{fig:architecture_left}). The terms Policy Decider, Policy Enforcer, and Data Monitor are introduced in this work to clarify the role of each rApp/xApp, as these last terms are generic.}

\subsection{vBS Policies}

We optimize the system operation over a time horizon of $t=1,\ldots, T$ rounds. For the DL, we define the set of the maximum allowed \emph{vBS transmission powers}, $\mathcal{P}^{\text{d}} \!=\! \big\{p_{i}^{\text{d}}, \,\forall i \! \in \!\{1, \ldots, P^\text{d}\} \big\}$, the set of highest eligible \emph{MCS}, $\mathcal{M}^{\text{d}} \!= \!\big\{m_{i}^{\text{d}}, \,\forall i \! \in \! \{1, \ldots, M^\text{d}\} \big\}$, and the set of maximum \emph{PRB ratio}, $\mathcal{B}^{\text{d}} \!=\! \big\{b_{i}^{\text{d}}, \,\forall i \!\in\! \{1, \ldots, B^\text{d}\} \big\}$, where $P^\text{d}$, $M^\text{d}$, and $B^\text{d}$ denote the number of possible transmission power, MCS, and PRB ratio levels in DL, respectively.\footnote{The MCS values are predetermined, and similarly, one can quantize the power and PRB ratio values; see, e.g. \cite{gomez2016srslte}.} The PRB ratio corresponds to the portion of the available PRBs the channel supplies, e.g., $b_{i}^{\text{d}} \!=\! 0.2$ leads to utilization of \SI{20}{\percent} ($10$ out of $50$ PRBs). The DL policy  for round $t$ is denoted with $ x_t^{\text{d}} \!\in\! \mathcal{P}^{\text{d}} \!\times\! \mathcal{M}^{\text{d}} \!\times\! \mathcal{B}^{\text{d}}$. Similarly, for the UL we introduce the sets $\mathcal{M}^{\text{u}} \!=\! \big\{m_{i}^{\text{u}}, \,\forall i \!\in\! \{1, \ldots, M^\text{u}\} \big\}$ and $\mathcal{B}^{\text{u}} \!=\! \big\{b_{i}^{\text{u}}, \,\forall i \!\in\! \{1, \ldots, B^\text{u}\} \big\}$, where $M^\text{u}$, $B^\text{u}$ are the available MCS and PRB ratio levels in UL\footnote{A maximum allowed UE transmission power is not defined since the users’ transmission power has less impact on the vBS power than the MCS and PRBs in the UL. However, it can be included in $x_t^\text{u}$ if deemed relevant for another application.} and denote with $ x_t^{\text{u}} \!\in\! \mathcal{M}^{\text{u}}$ the UL policy. Putting these together, the $t$-round threshold policy is:
\begin{align*}
    x_t = (x_t^{\text{d}}, x_t^{\text{u}}) \in \X, \,\, \text{where} \,\,  \mathcal{X} = \mathcal{P}^{\text{d}}\! \times\! \mathcal{M}^{\text{d}}\! \times\! \mathcal{B}^{\text{d}}\! \times\! \mathcal{M}^{\text{u}}\! \times\! \mathcal{B}^{\text{u}}.
\end{align*}

\subsection{Rewards \& Costs}

The first goal of the learner is to maximize the \emph{effective} DL and UL throughputs, which depend on the aggregate transmitted data and the backlog in each direction. In line with prior works (see \cite{ayala2021bayesian} and references therein), we use the \emph{utility} function:
\begin{equation} \label{eq:utility}
	U_t(x_t)=
	\log \left(1\!+\!\frac{R_t^{\text{d}}(x_t^{\text{d}})}{d_t^{\text{d}}}\right)\!+\log \left(1\!+\!\frac{R_t^{\text{u}}( x_t^{\text{u}})}{d_t^{\text{u}}}\right),
\end{equation}
where $d_t^{\text{d}}, d_t^{\text{u}} \!>\! 0$, with $U_t(x_t) \!=\! 0$ otherwise.
$R_t^{\text{d}}(\cdot)$ and $R_t^{\text{u}}(\cdot)$ denote the DL and UL \emph{transmitted data} during round $t$; and $d_t^\text{d}$ and $d_t^\text{u}$ are the respective backlogs, i.e., the \textcolor{black}{traffic} \emph{demands} during $t$. \revmajor{The logarithmic transformation balances the system utility across each stream (i.e., DL and UL), but we note that other mappings (e.g., linear) might be used to capture the specifics of different applications. \revminor{We have divided the transmitted data by the actual traffic demands in the respective stream (UL or DL), since the reward should naturally be defined w.r.t. the needs of the system.} Similarly, one can readily extend the utility function to capture various QoS metrics, e.g., by measuring only the throughput above a certain threshold.}
\revminor{We refrain from making assumptions about how $x_t^{\text{u}}, x_t^{\text{d}}$ affect the transmitted data, $R_t^{\text{d}}, R_t^{\text{u}}$; similarly, the traffic demands, $d_t^{\text{d}}, d_t^{\text{u}}$, are also considered unknown and can vary arbitrarily.\footnote{\revminor{Kindly refer to Sec. \ref{sec:evaluation} for details on their calculation during the evaluation of the proposed algorithms.}} In this black-box approach, each threshold policy $x_t$ (i.e., \textit{bandit arm}) yields a reward, which we calculate a posteriori, and corresponds to the reward of the respective bandit arm. The goal of our algorithms is to learn progressively which bandit arm leads to the highest possible reward.}

The second goal of the learner is to minimize the vBS energy costs. To that end, we introduce the time-varying \emph{power cost} function $P_t(x_t)$, which depends on {policy} $x_t$ in a possibly unknown fashion. Our decision to refrain from making assumptions about this function is rooted in the complexities involved in characterizing the power consumption and costs of such virtualized base stations \cite{jose-icc21}. \revmajor{Furthermore, this black-box approach allows us to capture a range of factors that might affect the consumed energy (e.g., retransmissions due to interference or time-varying electricity prices).}

Putting these together, the \emph{learner}'s criterion is the \emph{reward} function $ \tilde f_t\!: \mathbb \X \rightarrow \mathbb R$ defined as:
\begin{equation}\label{eq:reward}
\tilde f_t(x_t)={U_t(x_t)}-\delta {P_t(x_t)},
\end{equation}

\noindent where parameter $\delta\!>\!0$ is set by the network operator to tune the relative priority of utility and energy costs. Parameter $\delta$ serves as a metric transformation, enabling a meaningful scalarization of $U_t$ and $P_t$. Furthermore, we introduce, for technical reasons, the \emph{scaled} reward function $f_t:\X \rightarrow [0,1]$, since our learning algorithms (see Sec. \ref{sec:configuration_learning_for_adversarial_environments} and \ref{sec:universal_configuration_learning_through_an_experts_model}) operate on that interval. An easy-to-implement mapping that ensures this normalization is:

\begin{equation}
f_t(x_t)=\big( \tilde f_t(x_t)- \tilde f_{min} \big)/ \big(\tilde f_{max}-\tilde f_{min}\big).  \label{eq:mapping}
\end{equation}

Parameters $\tilde f_{min}$ and $\tilde f_{max}$ can be determined based on $\delta$, the min/max value of power cost function, the min/max vBS transmission power, PRB ratio, MCS and \textcolor{black}{traffic} \textcolor{black}{demands}.

\subsection{Environment}

We refer to {the ``external''} information, i.e., $\{c_t^{\text{d}},\,c_t^{\text{u}},\,d_t^{\text{d}},\,d_t^{\text{u}}\}_{t=1}^T$ as \emph{environment}, and it is responsible for shaping the function $f_t$. It is crucial to note that both reward components, $U_t$ and $P_t$, vary with time $t$, an effect that is attributed to several factors. First, the \textcolor{black}{traffic} demands, i.e., $d_t^{\text{d}}$ in DL and $d_t^{\text{u}}$ in UL, change, sometimes drastically, in every round $t$, e.g., in small-cell networks where user churn is high, which affects $U_t$, see \eqref{eq:utility}. The demands also impact the choice of MCS and PRB, leading to different processing times and, thus, different power costs. Second, the \textcolor{black}{channel qualities (i.e., CQIs)} in DL and UL, denoted as $c_t^{\text{d}}$ and $c_t^{\text{u}}$, respectively, might vary (in slow, fast, or mixed timescales), and this affects the transmitted data $R_t^{\text{d}}$ and $R_t^{\text{u}}$ (hence, $U_t$ changes even for fixed $x_t$) and the energy cost $P_t$ (low CQI induces more BBU processing \cite{jose-icc21}).\footnote{The operation cost of the vBS hosting platform is subject to variations in external computing loads (e.g. when co-hosting other services or other vBS/DUs), changes in the monetary cost (or availability) of the energy price, and so on.} 

Importantly, we consider the environment to be \emph{unknown} at the beginning of each scheduling round $t$. \revmajor{It is often challenging to predict the {traffic demands}, energy availability, {channel qualities}, wireless interference and other performance-related impairments that each vBS might encounter over the time window of several seconds that these threshold policies will be enforced. This, in turn, means that when we decide $x_t$ in each round, we do not have access to $f_t$; and this is in notable contrast to the typical real-time radio management solutions that require accurate context information. Our model is hence oblivious to this information and this renders our solution applicable to a range of practical scenarios, such as those involving highly volatile  environments and small cells where demands are non-stationary \cite{paschos-comm-mag}.}

\section{Policy Learning for Adversarial Environments}
 \label{sec:configuration_learning_for_adversarial_environments}

\subsection{Objectives \& Approach} 

\textcolor{black}{The goal of our rApp (see Policy Decider, Fig. \ref{fig:architecture})} is to find a sequence of \textcolor{black}{policies} $\{x_t\}_{t=1}^T$ that induce rewards approaching the cumulative reward of the single best \textcolor{black}{policy} (\emph{benchmark}). Formally, we employ the metric of \emph{static expected regret}:
\begin{equation}
\R_T= \max_{x\in \X	} \left\{ \sum_{t=1}^T f_t(x)\right\} - \bE \left[ \sum_{t=1}^T f_t(x_t)	\right], \label{eq:regret}
\end{equation}
\noindent where the first term is the aggregate performance of the benchmark (ideal) {policy} that can be selected only with knowledge of all future reward functions until $T$; and the second term measures the aggregate performance of the algorithm. The expectation in the second term is induced by any possible randomization \textcolor{black}{in $\{f_t\}_{t=1}^T$ and} in the selection of $\{x_t\}_{t=1}^T$ by the learner. Eventually, our objective is to devise a rule that decides the \textcolor{black}{policies} in such a way that the average regret, for any possible realization of {rewards} $\{f_t\}_{t=1}^T$, diminishes asymptotically to zero, i.e., $\lim_{T\rightarrow \infty}\R_T/T=0$. Importantly, we wish to ensure this condition: \emph{(i)} without knowing $f_t$ when deciding $x_t$, and \emph{(ii)} by observing only $f_t(x_t)$ when applying $x_t$, and not the complete function \textcolor{black}{$f_t(x), \forall x \in \X$,  as only \emph{one} policy $x_t$ in each round $t$ can be deployed to the vBS.}

The proposed scheme, named Bandit Scheduling for vBS (\texttt{BSvBS}), builds upon the \emph{Exp3} algorithm \cite{exp3_auer}, and its underlying idea is to learn the correct probability distribution $y_t^{\text{B}}$ \textcolor{black}{(B refers to \texttt{BSvBS})} from which we can sample $x_t$ for each round $t$:
\[
x_t \sim \mathbb{P}(x_t=x') = y_t^\text{B}(x'), \forall x' \in \mathcal{X}.
\] 
The distributions $\{y_t^\text{B}\}_{t=1}^{T}$ belong to the probability simplex:
\[
 \Y^\text{B}=\left\{ y^\text{B}\in [0,1]^{|\mathcal{X}|} \,\, \bigg \vert \,\,\sum_{x\in \mathcal{X}} y^\text{B}(x)=1	\right\}, \notag
\]
and are calculated in each round using the following explore / exploit rule:
\begin{equation}
    y_t^\text{B}(x) = \frac{\gamma}{|\mathcal{X}|} + (1-\gamma)\frac{w_t^\text{B}(x)}{\sum_{x' \in \mathcal{X}}{w_t^\text{B}(x')}},\,\,\, \forall \, x\in \mathcal{X}. \label{eq:update-distribution}
\end{equation}
\revmajor{This formula includes three components: (i) the exploration part, $ {1}/{|\mathcal{X}|}$ which selects a policy randomly, (ii) the exploitation part, ${w_t^\text{B}(x)}/{\sum_{x' \in \mathcal{X}}{w_t^\text{B}(x')}}$, which chooses a threshold policy based on its performance up until $t-1$, where the weight $w_t^\text{B}(x)$ tracks the reward of each policy $x \in \mathcal{X}$, and (iii) parameter $\gamma \in (0,1]$, which prioritizes the former (explore) or the latter part (exploit). }

For the latter, we employ the weight vector $w_t\!=\!\big(w_t(x)\!:\! x\!=\!1,\ldots, |\mathcal{X}|\big)$ that tracks the success of each tested policy, which is updated at the end of each round using:
\begin{equation}
    w_{t+1}^\text{B}(x)=w_t^\text{B}(x)\exp \left(\frac{\gamma \Phi_t^\text{B}(x)}{|\mathcal{X}|} \right),\,\,\,\forall\,x\in\mathcal{X}, \label{eq:weights}
\end{equation}
\revmajor{which assigns a probability exponentially proportional to the cumulative reward $\Phi_t^\text{B}(x)$, that accounts for the selection of each policy, namely:} 
\begin{equation}
\Phi_t^\text{B}(x)=
\begin{cases}
	f_t(x_{t}) / y_t^\text{B}(x_t), & \text{if} \,\, x = x_t,\\
	0, & \text{otherwise.}
\end{cases}
\label{eq:weighted-feedback}
\end{equation}
\revmajor{By dividing each observed reward, $f_t(x_{t})$ with the selection probability of the threshold-policy, $y_t^\text{B}(x_t)$, we ensure the conditional expectation of $\Phi_t^\text{B}(x)$ is the actual reward $f_t(x), \forall x \in \mathcal{X}$, meaning that $\Phi_t^\text{B}$ is an unbiased function estimator of the rewards \cite{cesa-bianchi_lugosi_2006}. Intuitively, this compensates the reward of thresholds that are unlikely to be chosen.} The steps of the learning scheme are summarized in Algorithm \ref{BSvBS}, which takes as input $\gamma$ and devises the ideal selection probability for each \textcolor{black}{policy} based on its expected reward.

\subsection{Optimality Guarantees}

The performance of Algorithm \ref{BSvBS} is \textcolor{black}{characterized} in the following lemma, which holds for any possible sequence of functions $\{f_t\}_{t=1}^T$:

\begin{lemma}\label{lemma_BSvBS}
Let $T > 0$ be a fixed time horizon. Set input parameter $\gamma = \min\left\{1, \sqrt{{|\mathcal{X}| \ln{|\mathcal{X}|}}/{\big((e-1)T\big)}}\right\}$. Then, running Algorithm 1 ensures that the expected regret is:

\begin{equation}
	\R_T \leq 2\sqrt{(e - 1)}\sqrt{T|\mathcal{X}|\ln{|\mathcal{X}|}} \label{eq:exp3_bound}
\end{equation}
\end{lemma}
\begin{proof}
\textcolor{black}{The proof follows by tailoring the main result of \cite{exp3_auer}, which provides an upper bound to \eqref{eq:regret}, namely:
\begin{align}
    \R_T \leq (e-1) \, \gamma  \,\max_{x\in \X	} \left\{ \sum_{t=1}^T f_t(x)\right\} + \frac{|\mathcal{X}|\ln{|\mathcal{X}|}}{\gamma}.\label{eq:lemma1_exp3}
\end{align}
The number of \emph{bandit arms} in our case corresponds to the eligible \textcolor{black}{policies}; hence it is equal to $|\X|$. Given that: \emph{(i)} the horizon $T$ can be known in advance, and \emph{(ii)} the rewards $f_t(x_t)$ for each chosen \textcolor{black}{policy} $x_t$ at round $t$ cannot be greater than 1 (due to the normalization described in Sec. \ref{sec:system_model}), we determine an upper bound $g$ of $\max_{x\in \X} \left\{ \sum_{t=1}^T f_t(x)\right\}$ equal to $T$, i.e., $g \!=\! T$. By choosing the suggested $\gamma$, \eqref{eq:lemma1_exp3} leads to \eqref{eq:exp3_bound}}.
\end{proof}

\setlength{\textfloatsep}{0pt}
\begin{algorithm}[t]\label{BSvBS} 
{\small{
 \nl \textbf{Parameters}: $\gamma = (0, 1]$\\ 
    \nl \textbf{Initialize}: at $t=1$, $w_{1}^\text{B}(x)\leftarrow 1, \,\,\forall x\in \mathcal X$\\
    \nl \For{$t=1,2,\ldots, T$}{
            \nl Define the probability $y_t^\text{B}(x), \,\,\forall x\in \mathcal X$ using \eqref{eq:update-distribution}.\\
            \nl Sample next \textcolor{black}{policy}: $x_t \sim y_t^\text{B}$.\\
            \nl Receive \& scale reward $f_t(x_t)$ using \eqref{eq:reward} and \eqref{eq:mapping}.\\
            \nl Calculate weighted feedback $\Phi_t^\text{B}(x), \,\,\forall x\in \mathcal X$ using \eqref{eq:weighted-feedback}.\\
            \nl Update $w_{t}^\text{B}(x), \,\,\forall x\in \mathcal X$ using \eqref{eq:weights}.
        }	
        \caption{{Bandit Scheduling for vBS (\texttt{BSvBS})}}
}}      
\end{algorithm} 

\subsection{Discussion for \texttt{BSvBS}}

Algorithm \ref{BSvBS}, which operates with bandit feedback, is guaranteed to achieve the same performance as the (unknown) single best \textcolor{black}{policy} without imposing any conditions on the system operation, \textcolor{black}{channel qualities}, or \textcolor{black}{traffic} demands; see Lemma \ref{lemma_BSvBS}. 

\revmajor{Regarding the overheads of this algorithm, \texttt{BSvBS} depends on the number of policies $|\mathcal{\mathcal{X}}|$ and the number of rounds $T$. Each round of the algorithm involves updating the probability distribution over the policies, see equation \eqref{eq:update-distribution}, which requires $\mathcal{O}(|\mathcal{X}|)$ time. Additionally, the algorithm updates the weights for each eligible threshold policy based on the reward, which again takes $\mathcal{O}(|\mathcal{X}|)$ time, see equations \eqref{eq:weights} and \eqref{eq:weighted-feedback}. Thus, for $T$ rounds, the time complexity is generally $\mathcal{O}(T|\mathcal{X}|)$. Also, its space complexity is $\mathcal{O}(|\mathcal{X}|)$, as it needs to store only the weights and the probabilities for each policy. In other words, the algorithm is both robust and lightweight in terms of implementation, especially compared to its main competitor, \texttt{BP-vRAN} \cite{ayala2021bayesian}, which has $\mathcal{O}(T^3)$ time complexity and $\mathcal{O}(T^2)$ space complexity.} Nevertheless, the robustness of \texttt{BSvBS} is achieved via a conservative approach that prevents the system from performing better when the conditions allow it. We tackle this issue in the following section.

\section{Universal \textcolor{black}{Policy} Learning through {a Meta-Learner}} \label{sec:universal_configuration_learning_through_an_experts_model}

\subsection{Modeling \& Challenges}

The analysis in Sec. \ref{sec:configuration_learning_for_adversarial_environments} demonstrates the effectiveness of the proposed adversarial scheme \textcolor{black}{in \emph{all} environments, whether challenging or easy}. However, in the latter case, alternative schemes that leverage the knowledge of the environment can achieve faster learning convergence \cite{ayala2021bayesian}. \textcolor{black}{\emph{Our goal here is to devise a meta-learning scheme that leverages multiple algorithms, each tailored to a specific environment, and chooses dynamically the optimal one.} This idea is leveraged in online learning \cite{NIPS2014_01f78be6}; however, to the best of the author's knowledge, it is hitherto unexplored for resource allocation in RAN.}

\setlength{\textfloatsep}{0pt}
\begin{algorithm}[t]\label{MetBS} 
	{\small{	
			\nl \textbf{Parameters}: $\eta = (0, 1]$\\ 
			\nl \textbf{Initialize}: at $t=1$, $w_{1}^\text{M}(j)\leftarrow 1$ and $h_0^{j, \text{S}} \leftarrow \emptyset, \,\,\forall j\in \mathcal{A}$\\
			\nl \For{$t=1,2,\ldots, T$}{
				\nl Define the probability $y_t^\text{M}(j), \,\,\forall j\in \mathcal{A}$ using \eqref{eq:update-distribution_MetBS}.\\
				\nl Sample algorithm $a^{i_t}$ according to: $a^{i_t} \sim y_t^\text{M}$.\\
				\nl Algorithm $a^{i_t}$ recommends \textcolor{black}{policy} $x_t^{i_t}$ based on $h_t^{i_t, \text{S}}$. \\ 
				\nl Receive \& scale reward $f_t(x_t^{i_t})$ using \eqref{eq:reward} and \eqref{eq:mapping}.\\
				\nl Calculate weighted feedback $\Phi^\text{M}(j), \,\,\forall j \in \mathcal{A}$ using \eqref{eq:weighted-feedback_MetBS}.\\
				\nl Update $w_{t}^\text{M}(j), \,\,\forall j \in \mathcal{A}$ using \eqref{eq:weights_MetBS}. \\
				\nl Sample $\xi_t$ using \eqref{eq:xi}. \\
				\nl \eIf{$\xi_t = 0$}{block feedback of algorithm $a^{i_t}$, i.e., $h_t^{i_t, \text{S}} \leftarrow h_{t-1}^{i_t, \text{S}}$.}{allow feedback of algorithm $a^{i_t}$, i.e., $h_t^{i_t, \text{S}} \leftarrow h_{t-1}^{i_t, \text{S}} \cup \big(x_t^{i_t}, f_t(x_t^{i_t})\big)$.} 
			}
			\caption{{Meta-learning for vBS (\texttt{MetBS})}}
	}}    
\end{algorithm} 

\revmajor{In practical terms, the implementation of such a meta-learning algorithm (Algorithm \ref{MetBS}) can be realized in the non-RT RIC, i.e., co-located with the Policy Deciders.} Namely,  we deploy $A$ {rApps, i.e.,} \emph{algorithms} $a^j, j \!\in\! \mathcal{A} \!=\! \{1, \ldots, A \} $, each associated with a set of \textcolor{black}{policies} $\mathcal{X}^j$; and another rApp for the \emph{meta-learner} that observes their performances over a time horizon of $t \!=\! 1,\ldots, T$ rounds \textcolor{black}{via the R1 interface (see Fig. \ref{fig:architecture})}. At a time $t$, an algorithm $a^j, j \!\in\! \mathcal{A}$ takes as input the \emph{full} history $h_{t}^{j} = \big\{\big(x_\tau^{j}, f_\tau(x_\tau^{j})\big) \big\}_{\tau=1}^{t-1}$ of its previously proposed {policies} and their respective rewards, and proposes a {policy} $x_t^j = a^j(h_t^j)$. The objective of the meta-learner is to find the best performing algorithm $a^{i^*}, $ $i^* \!\in\! \mathcal{A}$. \revmajor{The challenge lies in the fact that the algorithms are learning entities that update their proposed threshold policies based on bandit feedback, which in turn depends on whether they are selected by the meta-learner.} In other words, at round $t$, the meta-learner chooses one algorithm ${i_t} \in \mathcal{A}$, denoted as $a^{i_t}$, which, in turn, proposes one \textcolor{black}{policy} $x_t^{i_t} \in \mathcal{X}^{i_t}$ that is deployed in the vBS; and thus, reward $f_t(x_t^{i_t})$ is returned,\footnote{This is a natural approach for our problem setting, as each algorithm proposes possibly different \textcolor{black}{policies} at each round, but only the \textcolor{black}{policy} of one algorithm can be deployed to the vBS and return a reward.} cf. \eqref{eq:reward}. Lastly, $a^{i_t}$ updates its learning state by updating its history $h_t^{i_t} \leftarrow h_{t-1}^{i_t} \cup  \big(x_t^{i_t}, f_t(x_t^{i_t})\big)$. All other algorithms, i.e., $\forall j \!\in\! \mathcal{A}\!: j \neq i_t$, observe no feedback and do not update their learning state at time $t$. 

This downward spiral creates a challenging situation where the partial feedback reduces the learning capability of the meta-learner, which is further compounded by the limited chances of obtaining feedback for each \textcolor{black}{policy}. Without coordination between the meta-learner and the algorithms in the bandit setting, it is proven that the meta-learner will achieve linear regret, even if each of the algorithms obtains sub-linear regret if it were run on its own (and thus obtain feedback in every round) \cite{corralling, singla18}. To surmount this challenge, effective coordination between the algorithms and the meta-learner becomes essential. The approach we employ, inspired by the ideas presented in \cite{singla18}, aims to minimize the interaction required between the algorithms and the meta-learner. Other existing meta-algorithms such as \cite{corralling} and \cite{maillard-munos} require feeding unbiased estimates of rewards to the algorithms, meaning that the meta-learner has access to the rewards of the algorithms and can modify them; an assumption that we want to drop in our setting.

In our case, the meta-learner can allow or block the chosen algorithm $a^{i_t}$ from learning at round $t$ by sending a corresponding bit ($0$ or $1$). This means that each algorithm $a^j, j \in \mathcal{A}$ has access to \emph{sparse} history $h_{t}^{j, \text{S}} = \big\{\big(x_\tau^{j}, f_\tau(x_\tau^{j})\big) \,|\, \xi_\tau = 1\big\}_{\tau=1}^{t-1}$, where $\xi_\tau$ is a Bernoulli random variable, i.e., $\xi_\tau \sim \mathcal{B}(\rho_\tau)$, defined by the meta-learner. More precisely, with probability $\rho_t \in(0,1]$ at each round 
$t$, the meta-learner sends bit $1$, allowing the chosen algorithm $a^{i_t}$ to learn, i.e., update its history $h_t^{i_t, \text{S}} \leftarrow h_{t-1}^{i_t, \text{S}} \cup  \big(x_t^{i_t}, f_t(x_t^{i_t})\big)$; otherwise, $h_t^{i_t, \text{S}} \leftarrow h_{t-1}^{i_t, \text{S}}$. Obviously, it is true that if $\rho_t = 1,$ for $t=1,\ldots, T$, then $h_{t}^{j, \text{S}} \equiv h_{t}^{j}$. Intuitively, this prevents a situation where algorithms that initially find a good \textcolor{black}{policy}, but later experience a decline in performance, are continuously selected by the meta-learner over algorithms that explore more extensively in the early stages but achieve superior performance later. By choosing $\rho_t$ accordingly in every round $t$ (see the following analysis), all algorithms could observe feedback in an equal number of rounds (although the best-performing algorithms will be chosen more often) and thus have equal learning steps to improve their performance.

\subsection{Objectives \& Approach}
 
Following this rationale, the second proposed scheme, named \emph{Meta-Learning for vBS} (\texttt{MetBS}), builds upon \cite{singla18}. Due to \textcolor{black}{its similarity with} Algorithm \ref{BSvBS}, we elaborate next only on its most crucial and distinct steps. The concept lies in learning the sequence of distributions $\{y_t^\text{M}\}_{t=1}^{T}$ \textcolor{black}{(M refers to MetBS)}, which enables the selection of an algorithm ${i_t} \in \mathcal{A}$, denoted as $a^{i_t}$ at round $t$ based on the following explore-exploit criteria with parameter $\eta \in (0,1]$: 
\begin{equation}
    y_t^\text{M}(j) = \frac{\eta}{A} + (1-\eta)\frac{w_t^\text{M}(j)}{\sum_{j' \in \mathcal{A}}{w_t^\text{M}(j')}},\,\,\, \forall \, j\in \mathcal{A}. \label{eq:update-distribution_MetBS}
\end{equation}

Based on its history $h_t^{i_t, \text{S}}$ and its internal mechanism of using it (e.g., \texttt{BSvBS} uses \eqref{eq:update-distribution}), $a^{i_t}$ outputs a \textcolor{black}{policy} $x_t^{i_t} \in \mathcal{X}^{i_t}$. The meta-learner observes only the reward $f_t(x_t^{i_t})$ that $a^{i_t}$ produced, and thus, similarly to \texttt{BSvBS}, calculates an unbiased estimator for the rewards\footnote{We recall that no assumptions are made about the sequence of rewards $\{f_t\}_{t=1}^T$, which can even be chosen from an adversary, as described analytically in Sec. \ref{sec:system_model}.} of all the algorithms (even the unchosen ones):
\begin{equation}
    \Phi_t^\text{M}(j)=
    \begin{cases}
    	f_t(x_{t}^{i_t}) / y_t^\text{M}(i_t), & \text{if} \,\, j = i_t,\\
    	0, & \text{otherwise,}
    \end{cases}, 
    \forall j \in \mathcal{A}
    \label{eq:weighted-feedback_MetBS}
\end{equation}
The weights, which determine the meta-learner's choices in each $t$, are updated according to:
\begin{equation}
    w_{t+1}^\text{M}(j)=w_t^\text{M}(j)\exp \left(\frac{\eta \Phi_t^\text{M}(j)}{A} \right),\,\,\,\forall\,j\in \mathcal{A}. \label{eq:weights_MetBS}
\end{equation}
Before \texttt{MetBS} proceeds to the next round, it has the ability to block algorithm $a^{i_t}$ from acquiring feedback (i.e., learning) at this particular round $t$. Consequently, \texttt{MetBS} uses the following Bernoulli random variable to allow or block the feedback of $a^{i_t}$:
\begin{equation}
    \xi_t \sim \mathcal{B}\left(\frac{\eta}{A \,y_t^\text{M}(j) }\right), j = i_t.\label{eq:xi}
\end{equation}
More specifically, with probability $\rho_t = \eta/(A \,y_t^\text{M}(j)), \, j=i_t$ at each round $t$, the selected algorithm $a^{i_t}$ updates its learning state, while with the remaining probability, its feedback gets blocked. The selection of this random variable ensures that the feedback of each algorithm is allowed, on average, with constant probability $\rho = \eta/A$ over the whole horizon $T$. The analytical steps of this learning scheme are shown in Algorithm \ref{MetBS}.

It is crucial to stress that the regret of the meta-learner w.r.t. the best algorithm, cf. \eqref{eq:component1_regret_MetBS}, is uninformative on its own in the bandit setting. The reason can be attributed to the indirect association between rewards at any given time $t$ and the algorithms the meta-learner previously selected. The past selections define the current learning state of the algorithms, which, in turn, impacts the rewards \cite{maillard-munos}. Therefore, the evaluation should contain a comparison to an ideal policy that consistently selects the best algorithm, which obtains feedback in every $t$ and performs well with respect to the single best \textcolor{black}{policy}. Formally, we are interested in minimizing the regret of the meta-learner w.r.t. the single best \textcolor{black}{policy}, which is equal to:
\begin{align}
\R_T^\text{M}= \underbrace{\max_{x\in \X^{i^*}	} \left\{ \sum_{t=1}^T f_t(x)\right\}}_\text{best \textcolor{black}{policy}} - \underbrace{\bE\left[ \sum_{t=1}^T f_t\big(a^{i_t}(h^{i_t, \text{S}}_t)\big)\right]}_\text{meta-learner}.\label{eq:all_regret_MetBS}
\end{align}

The aggregate reward of the best algorithm $a^{i^*}$ achieved until round $t$ is:
\begin{align}
\max_{j \in \mathcal{A}} \left\{ \sum_{t=1}^T \bE \bigg[f_t\big(a^j(h^{j, \text{S}}_t)\big)\bigg]\right\} \equiv \sum_{t=1}^T \bE \bigg[f_t\big(a^{i^*}(h^{i^*, \text{S}}_t)\big)\bigg].\label{eq:agg_reward_best_algo}
\end{align}

We add and subtract \eqref{eq:agg_reward_best_algo} from \eqref{eq:all_regret_MetBS}, and we derive:

\begin{align}
\R_T^\text{M}= \R_T^{\text{M}_1} +  \R_T^{\text{M}_2}, \label{eq:all_split_regret_MetBS}
\end{align}
where $\R_T^{\text{M}_1}$ corresponds to the regret of the meta-learner with respect to the best algorithm:
\begin{align}
    \R_T^{\text{M}_1}= \underbrace{\sum_{t=1}^T \bE \bigg[f_t\big(a^{i^*}(h^{i^*, \text{S}}_t)\big)\bigg]}_\text{best algorithm} - \underbrace{\bE\left[ \sum_{t=1}^T f_t\big(a^{i_t}(h^{i_t, \text{S}}_t)\big)\right]}_\text{meta-learner}, \label{eq:component1_regret_MetBS}
\end{align}
and $\R_T^{\text{M}_2}$ corresponds to the regret of the best algorithm w.r.t. to the best \textcolor{black}{policy}:
\begin{align}
    \R_T^{\text{M}_2}=
    \underbrace{\max_{x\in \X^{i^*}	} \left\{ \sum_{t=1}^T f_t(x)\right\}}_\text{best \textcolor{black}{policy}} - \underbrace{\sum_{t=1}^T \bE \bigg[f_t\big(a^{i^*}(h^{i^*, \text{S}}_t)\big)\bigg]}_\text{best algorithm}.
 \label{eq:component2_regret_MetBS}
\end{align}

If $a^{i^*}$ had access to its full history $h_t^{i^*}$, we denote as $\beta^{i^*} \in [0,1]$ the exponent of the upper bound of its regret, namely\footnote{\textcolor{black}{For instance, if \texttt{BSvBS} is the best algorithm $a^{i^*}$, then $\beta^{i^*} = 1/2$, see Lemma \ref{lemma_BSvBS}.}}:
\[
\max_{x\in \X^{i^*}} \left\{ \sum_{t=1}^T f_t(x)\right\} - \sum_{t=1}^T \bE \bigg[f_t\big(a^{i^*}(h^{i^*}_t)\big)\bigg] \leq \mathcal{O}(T^{\beta^{i^*}}).
\]
However, in the considered analysis, it has access to its partial history $h_t^{i^*, \text{S}}$. For proving a non-trivial upper bound on $\R_T^{\text{M}}$ in this case, the best performing algorithm $a^{i^*}$ should satisfy the following:
\begin{align}
    \max_{x\in \X^{i^*}} \left\{ \sum_{t=1}^T f_t(x)\right\} - \sum_{t=1}^T \bE \bigg[f_t\big(a^{i^*}(h^{i^*, \text{S}}_t)\big)\bigg] \leq \mathcal{O}\bigg(\frac{(\rho T)^{\beta^{i^*}}}{\rho}\bigg),\label{eq:smooth_dynamics}
\end{align}
where $\rho = \eta / A$, as defined beforehand. A rich class of online learning algorithms, including Exp3 (and thus, \texttt{BSvBS}), satisfy \eqref{eq:smooth_dynamics}, \textcolor{black}{which, in turn, quantifies the robustness of an online learning algorithm w.r.t. the sparsity of the history} \cite{singla18}.

\subsection{Optimality Guarantees}
The performance of Algorithm \ref{MetBS} is captured by the following lemma:

\begin{lemma}\label{lemma_MetBS}
Let $T > 0$ be a fixed time horizon, and assume the best algorithm, $a^{i^*}$, satisfies \eqref{eq:smooth_dynamics} with $\beta^{i^*}$. Set input parameter $\eta = \Theta\big(T^{-\frac{1-\beta}{2-\beta}} A^\frac{1-\beta}{2-\beta}(\log A)^{\frac{1}{2} \mathbbm{1}_{\{\beta = 0\}}}\big)$, where $\beta \geq \beta^{i^*}$. Then, running Algorithm \ref{MetBS} ensures that the expected regret is sub-linear:

\begin{equation}
    \R_T^\text{M} \leq  \mathcal{O}\big(T^{\frac{1}{2-\beta}} A^\frac{1}{2-\beta}(\log A)^{\frac{1}{2} \mathbbm{1}_{\{\beta = 0\}}}\big)\label{eq:MetBS_bound}
\end{equation}
\end{lemma}

\begin{proof}
The proof follows by tailoring the main result of \cite{singla18}; we therefore provide a brief but sufficient explanation. \textcolor{black}{By applying Lemma \ref{lemma_BSvBS}}, \eqref{eq:component1_regret_MetBS} gives:

\begin{align}
    \R_T^{\text{M}_1} \leq c \, \eta  \,T + \frac{A \log A}{\eta},\label{eq:lemma2_exp3}
\end{align}

\textcolor{black}{where $c >0$ is a constant.} Adding \eqref{eq:smooth_dynamics} and \eqref{eq:lemma2_exp3}, results in: 
\begin{align}
    \R_T^\text{M} \leq \mathcal{O}\big(\eta \, T + \frac{A \log A}{\eta} + \frac{T^{\beta^{i^*}} A^{1-\beta^{i^*}}}{\eta^{1-\beta^{i^*}}} \big).\label{eq:lemma2_regret_MetBS}
\end{align}
Setting $\eta \sim T^{-z}$ and finding the $z$ that minimizes the power of $T$ in \eqref{eq:lemma2_regret_MetBS}, leads to \eqref{eq:MetBS_bound}.
\end{proof}

\subsection{Discussion for \texttt{MetBS}}

When interacting with \emph{learning algorithms} in the \emph{bandit} setting, Algorithm \ref{MetBS} is guaranteed to achieve the same performance as the best algorithm if it ran on its own (and thus, acquiring feedback in every round). Hence, \texttt{MetBS} attains reward as the (unknown) single best algorithm without making assumptions for the environment (see Lemma \ref{lemma_MetBS}). This accomplishment is made possible through minimum coordination between the meta-learner and the algorithms, as described in lines 10-11 of Algorithm \ref{MetBS}. 

\revmajor{In terms of implementation, \texttt{MetBS} can be implemented as another rApp, which also facilitates its coordination with the co-located rApps implementing the different algorithms; see also Fig. \ref{fig:architecture_right}. Regarding its overheads, due to its similarity with  \texttt{BSvBS}, its complexity depends on the number of algorithms that it chooses from, i.e., $\mathcal{O}(T|\mathcal{A}|)$ for $T$ rounds. However, as it chooses between different algorithms (where each of them selects policies and has its own complexity), the overall time complexity of \texttt{MetBS} depends on the worst-case scenario of the most time-complex algorithm. Similarly, its space complexity is equal to $\mathcal{O}(|\mathcal{A}|)$; however, an important factor is the complexity of the algorithms that it chooses from, and especially, the most space-complex algorithm.}

\section{Performance Evaluation} \label{sec:evaluation}

\subsection{Experimental Setup \& Scenarios}

\revmajor{The solutions are assessed under different traffic and environment scenarios using our recent publicly-available dataset \cite{ayala2021bayesian} with power consumption and throughput measurements from an O-RAN compatible testbed. This experimental setup includes a vBS and a UE\footnote{\revmajor{The usage of one UE is not limiting for our study, since the algorithm devises each vBS's thresholds based on the average (across users) throughput and energy, and the average CQI and traffic, i.e., the UE emulates the load of multiple users.}}, implemented as srseNB and srsUE from the srsRAN suite \cite{gomez2016srslte}.} The RUs of the vBS and UE are composed of an Ettus Research USRP B210, and their BBUs and near-RT RICs are implemented on general-purpose computers (Intel NUC BOXNUC8I7BEH). The power consumption of the BBU and RU is measured with the GW-Instek GPM-8213. A \SI{10}{\mega\hertz} band is selected, supplying a maximum capacity of approximately \SI{32}{Mbps} and \SI{23}{Mbps} for the downlink and uplink operation, respectively. \revmajor{The non-RT threshold policies are calculated in a programming language, emulating the operation of rApps; the real-time scheduling decisions are made by the default srsRAN scheduler that has been amended to comply with the MCS, PRB, and power thresholds that are provided to them in each round.}

The dataset contains \SI{32797} measurements for $|\mathcal{X}| \!=\! 1080$ {policies} corresponding to  $\mathcal{B}^{\text{d}} \!=\! \{0, 0.2, 0.6, 0.8, 1\}$,  $\mathcal{B}^{\text{u}} \!=\! \{0.01, 0.2, 0.4, 0.6, 0.8, 1\}$, $\mathcal{M}^{\text{d}}\!=\! \{0, 5,11, 16, 22, 27\}$,  $\mathcal{P}^{\text{d}} \!=\! \{3\}$\footnote{The DL transmission power is determined through the transmission gain of the USRP implementing the BS. The RU of the testbed is equipped with a fixed power amplifier that consumes \SI{3}{\watt} and a variable attenuator for power calibration. To account for this limitation, the dataset power measurements are post-processed using linear modeling \cite{ayala2021bayesian}.} and $\mathcal{M}^{\text{u}}\! =\! \{0, 5, 9, 14, 18, 23\}$. The random perturbations in this setup, as explained in Sec. \ref{sec:system_model}, emanate due to time-varying UL and DL demands, $\{d_t^{\text{u}},d_t^{\text{d}}\}_{t=1}^T$, measured in \si{Mbps}, and time-varying CQIs, $\{c_t^{\text{u}}, c_t^{\text{d}}\}_{t=1}^T$, which are dimensionless. 
\revminor{The transmitted data, $\{R_t^{\text{u}},R_t^{\text{d}}\}_{t=1}^T$, are calculated by multiplying the values of $\mathcal{B}^{\text{d}}$ ($\mathcal{B}^{\text{u}}$) with the transport block size (TBS); the latter is determined by mapping the $\mathcal{M}^{\text{d}}$ ($\mathcal{M}^{\text{u}}$) with the TBS index \cite{3gpp_36_213}. W.l.o.g., we have assumed 50 PRBs.}
The power cost function is set to $P_t(x_t) \!=\! V_t$, where $V_t$ is the total power consumed by the vBS, and the utility function as stated in \eqref{eq:reward}. W.l.o.g., we scale both components of the reward function to $[0, 1]$ and choose $\delta=1.5$ to prioritize the power consumption unless stated otherwise.  We set $\gamma\!=\!0.29$ for \texttt{BSvBS} and $\eta\!=\!0.04$ for \texttt{MetBS}, and use $T\!=\!50k$. All results are averaged over 10 independent experiments.

\revmajor{For the ensuing analysis, we assess three scenarios which represent a static environment (fixed, time-invariant parameters); a stationary stochastic environment (i.i.d. parameters); and an adversarial scenario. The latter, clearly, is an extreme case (e.g., can appear under high mobility conditions, heavy interference or attacks) that we use to demonstrate the robustness of the learning algorithms. On the other hand, the first two scenarios are in line with those typically considered by prior benchmarks, e.g., \cite{ayala2021bayesian, edgebol21}. In detail:}

\noindent $\bullet$ \textbf{\textit{Scenario A (static)}}: the demands and CQIs take the highest possible values according to our testbed, i.e., $d_t^{\text{d}} = 32,\, d_t^{\text{u}} = 23$, $c_t^{\text{d}} = 15,\, c_t^{\text{u}} = 15$.

\noindent $\bullet$ \textbf{\textit{Scenario B (stationary)}}: the demands and CQIs  \revmajor{are drawn randomly from fixed uniform distributions in each round, i.e., $d_t^{\text{d}} \sim \mathcal{U}(29, 32)$, $d_t^{\text{u}} \sim \mathcal{U}(20, 23)$, $c_t^{\text{d}},\, c_t^{\text{u}} \sim \mathcal{U}(1, 3)$, where $\mathcal U(a,b)$ denotes the uniform distribution over the interval $[a,b]$.}

\noindent $\bullet$ \textbf{\textit{Scenario C (adversarial)}}: the demands and CQIs are drawn randomly in a \emph{ping-pong} way; namely, in \emph{odd} rounds according to $d_t^{\text{d}} \sim \mathcal{U}(29, 32)$, $d_t^{\text{u}} \sim \mathcal{U}(20, 23)$, $c_t^{\text{d}},\, c_t^{\text{u}} \sim \mathcal{U}(13, 15)$, and in \emph{even} rounds from $d_t^{\text{d}},\, d_t^{\text{u}} \sim \mathcal{U}(0.01, 1)$, $c_t^{\text{d}}, c_t^{\text{u}} \sim \mathcal{U}(1, 3)$. \footnote{CQI 13 and 15 correspond to SNR of \SI{25}{\decibel} and \SI{29}{\decibel}, while CQI 1 and 3 to SNR of \SI{1.95}{\decibel} and \SI{6}{\decibel}, respectively.} \textcolor{black}{We note note that the learner does not have access to this information, and is oblivious to when the switches happen.}

\noindent Scenario C resembles dynamic environments, where the parameters might change drastically every round. This corresponds to the most challenging-to-learn \emph{adversarial} schemes in regret analysis, cf. \cite{bubeck_bandits}. Clearly, an algorithm that performs well under this case is guaranteed to perform well in all other scenarios. In the sequel, we use these scenarios to explore the convergence of the proposed learning and meta-learning algorithms, and compare them with selected state-of-the-art competitors in terms of \textit{(i)} regret, \textit{(ii)} vBS power savings, and \textit{(iii)} inference time.

\begin{figure}[!t]
    \centering
    \subfigure[]{\centering \hspace*{-1cm}\label{fig:BSvBS-static-15_15_32_23-100k_regret}\includegraphics[width=0.82\linewidth]{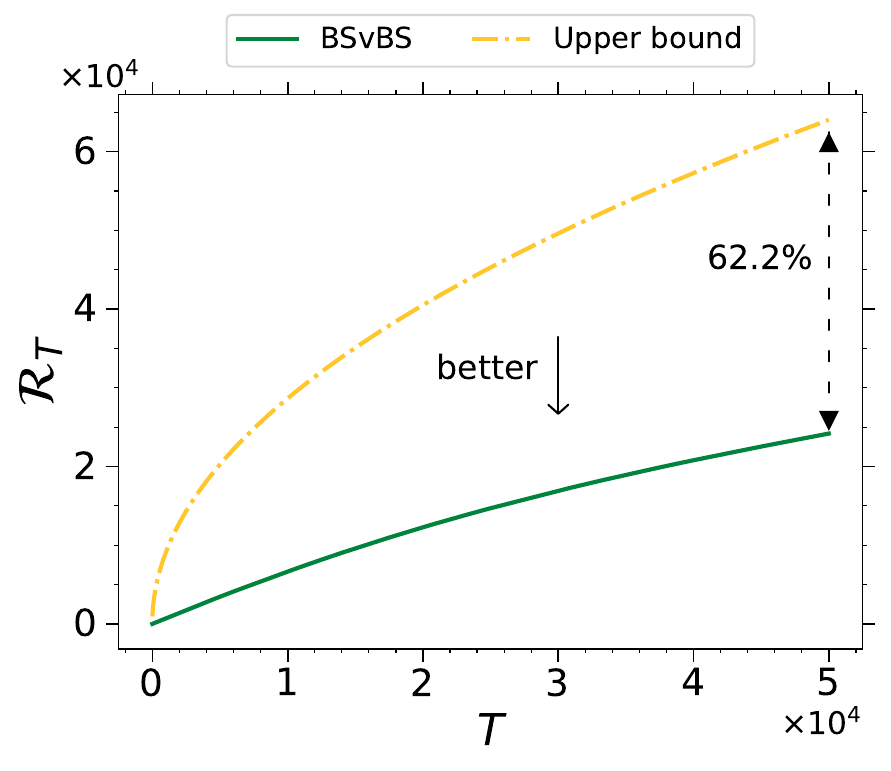}}
    \subfigure[]{\centering \label{fig:BSvBS-static-15_15_32_23-100k_heatmap}\includegraphics[width=0.85\linewidth]{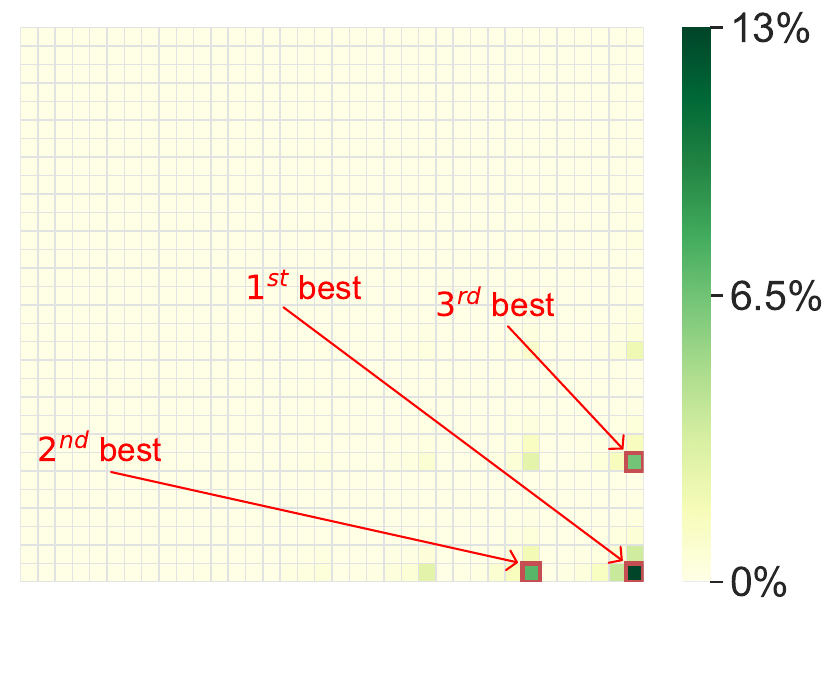}}
    \caption{\textbf{(a)} $R_T$ achieved from \texttt{BSvBS} in Scenario A (static) and its upper bound; \textbf{(b)} heatmap for the choices of \texttt{BSvBS} in Scenario A, showing the probability that each policy is chosen at $t=50k$.}
    \label{fig:static_BSvBS}
\end{figure}

\subsection{Static \& Stationary Scenarios}

\begin{figure*}[!t]
    \centering
    \subfigure[]{\label{fig:BSvBS-15_15_32_23-50k-static_MCS}\includegraphics[scale=0.34]{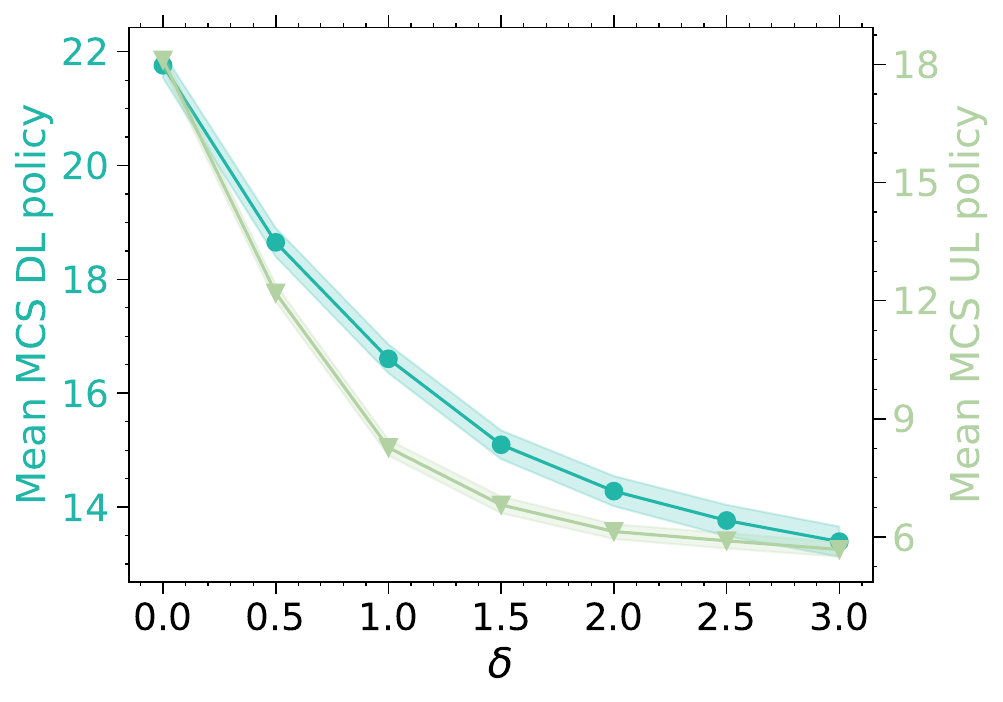}}
    \subfigure[]{\label{fig:BSvBS-15_15_32_23-50k-static_airtime}\includegraphics[scale=0.34]{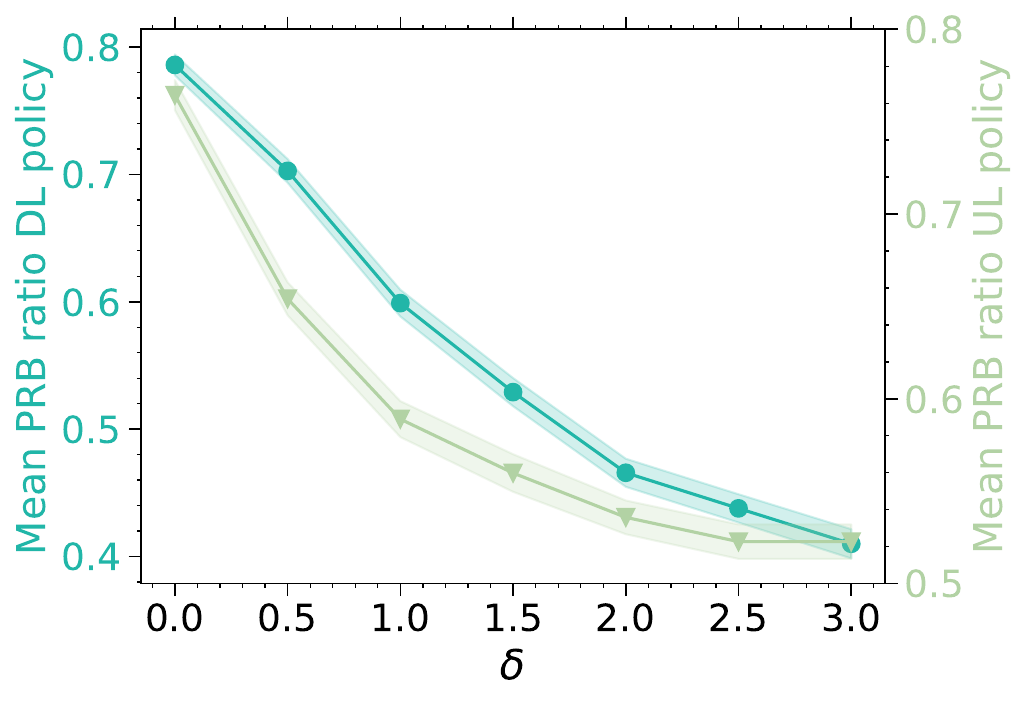}}
    \subfigure[]{\label{fig:V2_BSvBS-15_15_32_23-50k-static_power_utility.pdf}\includegraphics[scale=0.34]{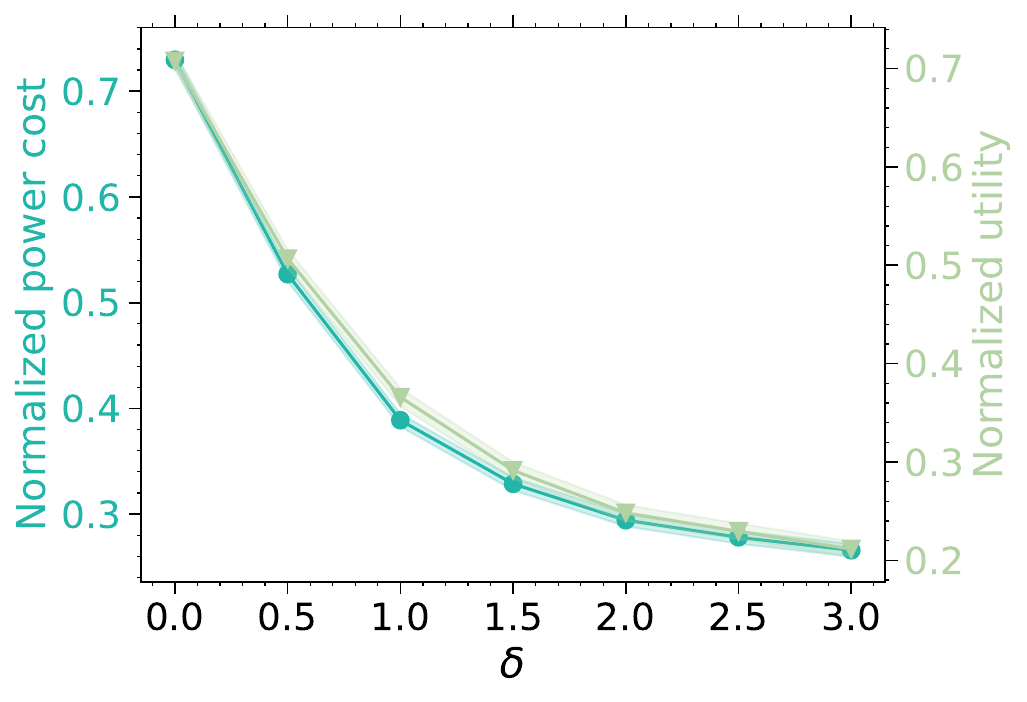}}  
    \caption{Scenario A (static) for \texttt{BSvBS}: \textbf{(a)} MCS in DL (left) / UL (right); \textbf{(b)}  PRB ratio in DL (left) / UL (right); \textbf{(c)} power (left) and utility (right) w.r.t. $\delta$, with 0.95-CI. \textcolor{black}{In each plot, the blue and green lines correspond to the left and right y-axis, respectively.}}
    \label{fig:static_config_choices}
\end{figure*}

Fig. \ref{fig:BSvBS-static-15_15_32_23-100k_regret} shows the expected regret in Scenario A when prioritizing the utility function (small $\delta$). The attained regret is sub-linear and \SI{62.2}{\percent} smaller than the upper bound (which is itself sub-linear), cf. \eqref{eq:exp3_bound}. To complement the analysis, Fig. \ref{fig:BSvBS-static-15_15_32_23-100k_heatmap} shows a grid with \SI{1080}{} cells, each mapping a different \textcolor{black}{policy}. The cells are colored based on the probability \texttt{BSvBS} selects each \textcolor{black}{policy} at $t=50k$, where darker colors indicate higher probabilities. The red squares indicate the \textcolor{black}{three} best \textcolor{black}{policies} chosen \SI{25}{\percent} of the rounds, where the top-performing one is selected twice as frequently. This outcome can be attributed to the small $\delta$, which favors the \textcolor{black}{policy} with the highest MCS and PRB ratio in both DL and UL, as the demands and CQIs are high.
For the second and third-best \textcolor{black}{policies}, the MCS in UL and DL take the highest values, except for the PRB ratios, which are fixed at 0.8, namely, the second-best UL and DL PRB ratios.

Fig. \ref{fig:BSvBS-15_15_32_23-50k-static_MCS} and \ref{fig:BSvBS-15_15_32_23-50k-static_airtime} delineate the effect of $\delta$ on the MCS DL/UL, and PRB ratio DL/UL, respectively (i.e., the chosen policies), for the static scenario. The solid lines in the plots represent the mean values averaging $100$ rounds after running \texttt{BSvBS} for $t=50k$ rounds, and the shadowed areas are the 0.95-confidence intervals. \textcolor{black}{Moreover, the blue and green lines correspond to the left and right y-axis, respectively.} We observe that smaller $\delta$ leads to higher MCS and PRB ratio choices in DL and UL. This is justified by the high CQI values considered in this scenario, as they enable using higher MCS, which allows more data transmission and larger decoding computational load \cite{rost15}. Furthermore, larger $\delta$ in Scenario A effectuates the selection of lower MCS and PRB values in order for the vBS to save resources by diminishing the turbo decoding iterations.

Similarly, Fig. \ref{fig:V2_BSvBS-15_15_32_23-50k-static_power_utility.pdf} illustrates the impact of $\delta$ on the reward function, where its two components are normalized, see \eqref{eq:reward}. Higher $\delta$ boosts the usage of policies that minimize the consumed power, forcing the utility function to decrease, whereas lower $\delta$ leads to \textcolor{black}{policies} that maximize the utility but increase the power consumption. Values $\delta>2$ have less effect on the power and utility functions, as there is a limit in the consumed power that can be saved. 

\begin{figure}[!t]
    \centering
    \centering\label{fig:V2_BSvBS-3_3_32_23-50k-stationary_regret}\includegraphics[width=0.85\linewidth]{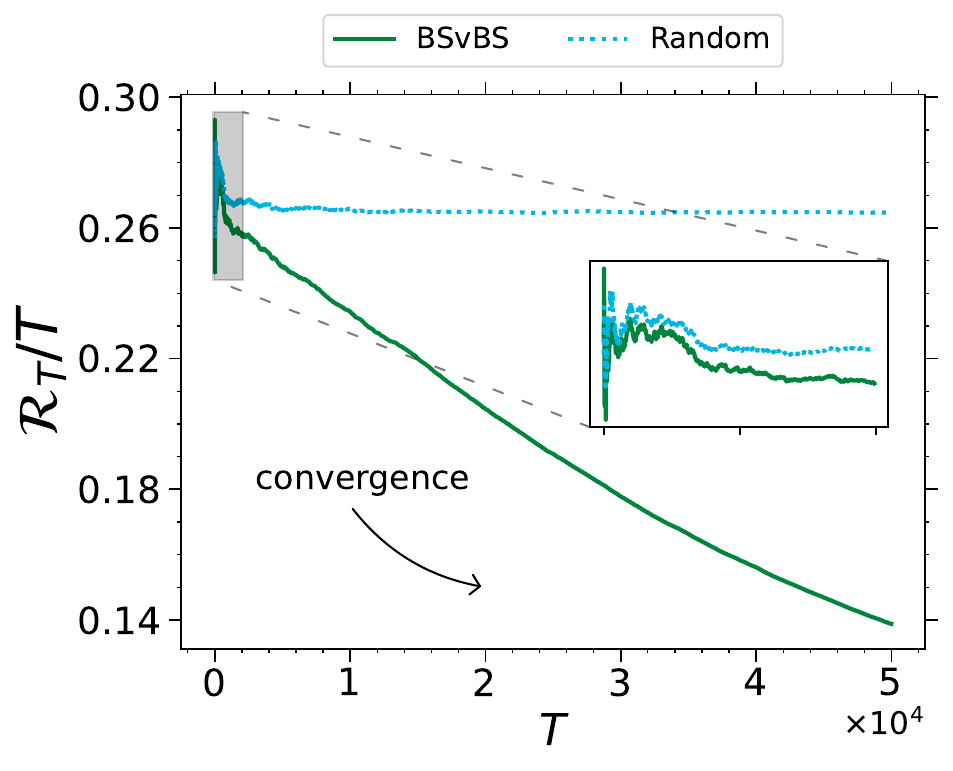}
    \caption{$R_T/T$ for \texttt{BSvBS} in Scenario B (stationary), together with \texttt{Random}, a naive algorithm that selects policies randomly.}
    \label{fig:stationary_BSvBS}
\end{figure}

Fig. \ref{fig:V2_BSvBS-3_3_32_23-50k-stationary_regret} depicts the average regret over time for stationary Scenario B, which converges towards zero as time elapses. We also plot the average regret of a typical benchmark that randomly selects \textcolor{black}{policies} with equal probability; we call this benchmark \texttt{Random}. \texttt{BSvBS} explores \textcolor{black}{policies} with probability \SI{29}{\percent} (since $\gamma=0.29$) and exploits the best-performing ones with probability \SI{71}{\percent}. Therefore, in the first $800$ rounds, \texttt{BSvBS} obtains similar regret as the benchmark algorithm, but their performance difference grows gradually, reaching  \SI{33.3}{\percent} in round $t=50k$, as \texttt{BSvBS} opts for the best-performing \textcolor{black}{policies} with higher probability at latter stages.

\noindent\textbf{\textit{{Key takeaways}}}: \textit{(i)} The measured regret is sub-linear in static and stationary scenarios and substantially smaller (up to \SI{62.2}{\percent}) than the theoretical bound. \textit{(ii)} The network can adjust $\delta$ to trade certain power consumption with commensurate losses in utility; yet, increasing $\delta$ more than a specific value ($\delta\!=\!2$ in our case) does not provide further substantial savings.

\begin{figure}[!t]
    \centering
    \subfigure[]{\label{fig:gap_prior_work_regret}\includegraphics[width=0.81\linewidth]{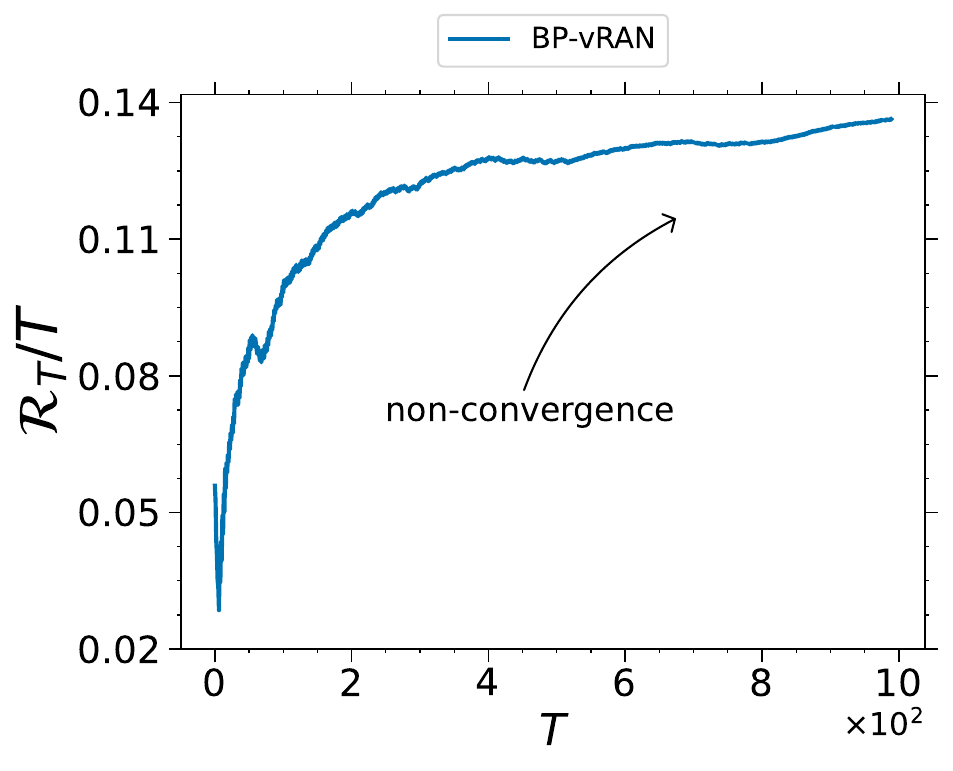}}
    \subfigure[]{\label{fig:gap_prior_work_choices}\includegraphics[width=0.79\linewidth]{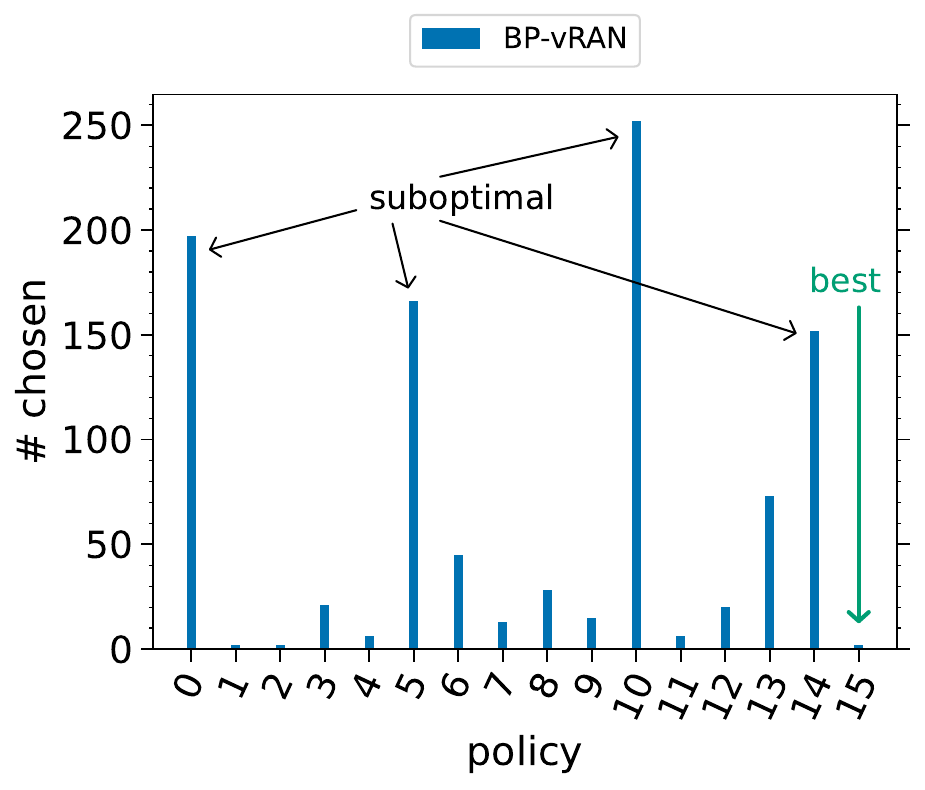}}
    \caption{\texttt{BP-vRAN} executed for $T=1000$ rounds in dynamic Scenario C, in a subset of the \textcolor{black}{policy} space: \textbf{(a)} $R_T/T$; \textbf{(b)} number of times each policy is chosen.}
    \label{fig:gap_prior_work}
\end{figure}

\subsection{Gap in Prior Work}

The primary objective is to showcase how state-of-the-art techniques perform inadequately in \textcolor{black}{challenging} environments. To delineate this effect, we focus on a smaller set of policies, i.e., $|\mathcal{M}_{d}|\!=\! |\mathcal{M}_{u}|\! =\! |\mathcal{B}_{u}| \!=\! |\mathcal{B}_{d}| \!=\! 2$ and $|\mathcal{P}_{d}| \!=\! 1$, yielding $|\mathcal{X}| \!=\! 16$ \textcolor{black}{policies}. The performance of the \texttt{BP-vRAN} algorithm \cite{ayala2021bayesian}, which constitutes, to the best of the authors' knowledge, the only existing work designed to configure such threshold policies in vBS, is assessed in adversarial Scenario C. \texttt{BP-vRAN}, which is based on the seminal GP-UCB algorithm \cite{Krause-GPUCB-ICML10}, models the \textcolor{black}{traffic} demands and CQIs as \emph{context}, which are observed before the policy is decided. Given that the context directly impacts the selection of policies, it will be shown how abrupt changes in CQI values and traffic demand deteriorate the algorithm performance. \textcolor{black}{We present an example where
the context differs between its observation and application to
the system.} This case appears quite often in practice, given that the rounds of reference are of several seconds. For the plots in this section, the reward function $f_t(x_t)$ is unbounded.\footnote{When \texttt{BSvBS} is depicted in the same plot as \texttt{BP-vRAN}, the reward function of \texttt{BP-vRAN} is scaled too.}

As indicated in Fig. \ref{fig:gap_prior_work_regret}, the average regret in the adversarial Scenario C does not decrease (in fact, it increases) after $T=1k$ rounds, which is more than 33$\times$ of the advertised convergence time. This happens because the algorithm takes decisions in each $t$ by assuming perfect knowledge of $f_t$, which might take arbitrary values depending on the  \textcolor{black}{environment}. Clearly, due to the system's volatility, the policy for each $t$ should be selected based on past values $\{f_\tau(x_\tau)\}_{\tau=1}^{t-1}$; yet, as Fig. \ref{fig:gap_prior_work_choices} corroborates, \texttt{BP-vRAN} selects sub-optimal \textcolor{black}{policies} for most rounds and fails to explore efficiently even this small space. 

\begin{figure}[!t]
    \centering
    \subfigure[]{\label{fig:V2_BSvBS-dynamic-50k_regret}\includegraphics[width=0.79\linewidth]{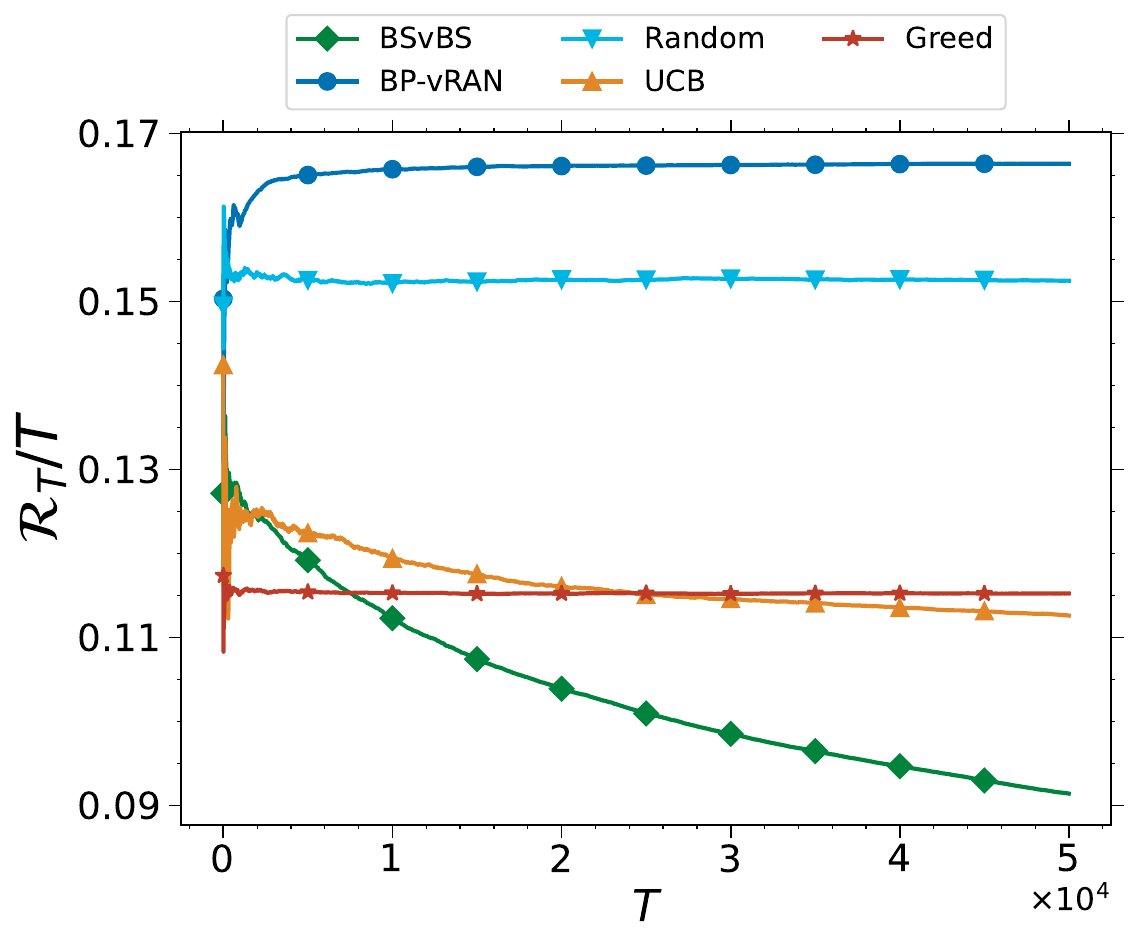}} 
    \subfigure[]{\label{fig:V2_power_gap-dynamic-50k}\includegraphics[width=0.78\linewidth]{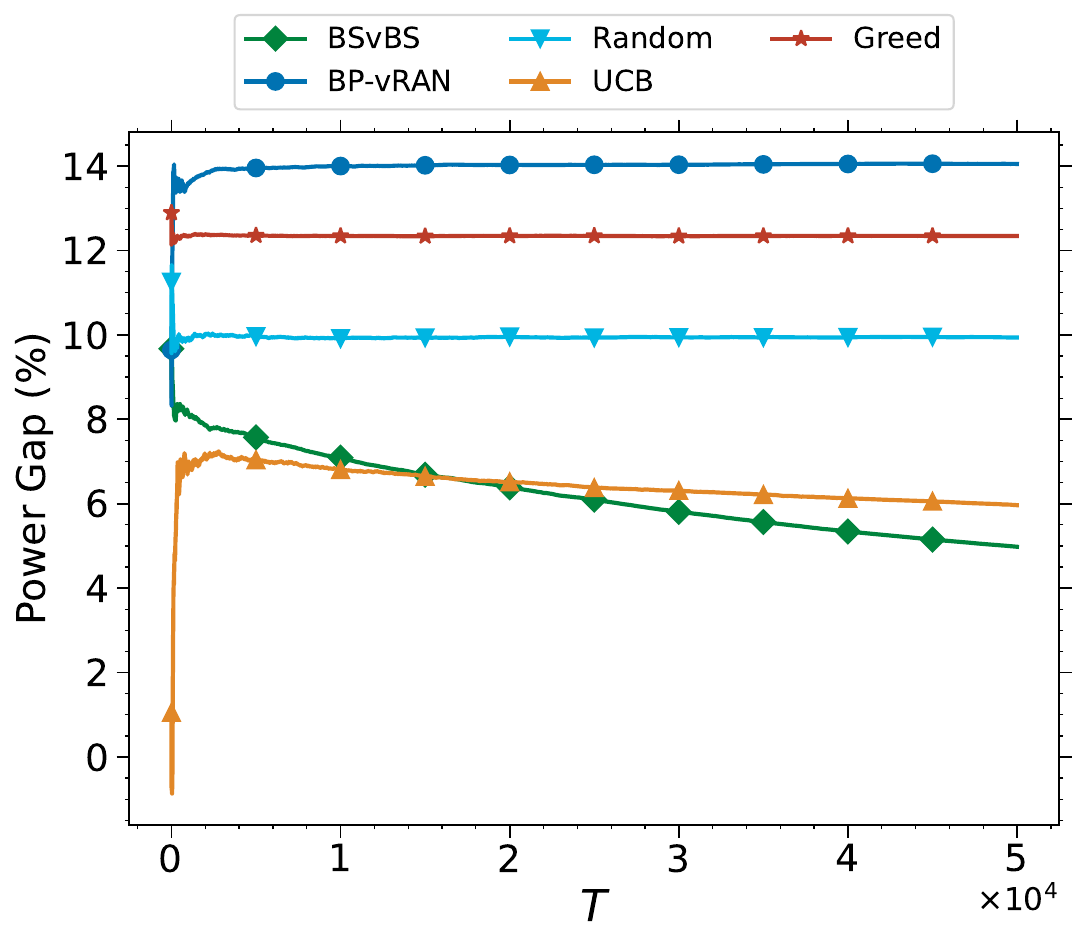}} 
    \caption{\revmajor{Comparison of \texttt{BSvBS} with several competitors in adversarial Scenario C: \textbf{(a)} $R_T/T$; \textbf{(b)} power saving of each algorithm with respect to the ideal-minimum energy of the benchmark.}} 
    \label{fig:dynamic_BSvBS}
\end{figure}

\subsection{Evaluation of the Bandit Algorithm}

\revmajor{Fig. \ref{fig:V2_BSvBS-dynamic-50k_regret} compares the average regret over time of \texttt{BSvBS} for Scenario C, in relation to several competitor algorithms, namely: the \texttt{BP-vRAN}; a naive algorithm that selects thresholds uniformly randomly (\texttt{Random}); the classical \texttt{UCB} algorithm that is designed for stationary environments \cite{auer_ucb}; and a greedy algorithm that prioritizes exploitation (\texttt{Greed}, selects the best solution found until now) \cite{sutton_RL}. We consider $T=50k$ rounds and use the complete {policy} space (i.e., $|\mathcal{X}|=1080$), and all results are averaged over $10$ independent experiments. We observe that \texttt{BSvBS} is superior, acquiring 45.1\% less regret w.r.t \texttt{BP-vRAN}, and 22\% less w.r.t \texttt{Greed} and \texttt{UCB} at $t = 50k$. It is worth noting that \texttt{Random} performs better than \texttt{BP-vRAN} in this case, by approximately 9\%. 

In Fig. \ref{fig:V2_power_gap-dynamic-50k}, we present the vBS power gains that each algorithm achieves in the same scenario, w.r.t. the ideal-minimum-energy of the benchmark, where the power consumption of the idle user is subtracted. It can be seen that with \texttt{BSvBS}, the network operator can save up to 64.5\% of energy if the algorithm runs for $t = 50k$ rounds in contrast to \texttt{BP-vRAN}. Moreover, it can be seen that \texttt{UCB} also chooses policies that allow for saving energy, but again, attains less energy saving than \texttt{BSvBS}. These plots also showcase that the \texttt{Greed} algorithm, which does not explore new policies, is not competitive and is stuck in exploiting sub-optimal policies (straight line in the regret plot).}

Another key advantage of \texttt{BSvBS} is its low inference time, i.e., the time to deduce a \textcolor{black}{policy} in each round. Fig. \ref{fig:V2-inference time} exhibits the average inference time and compares it with \texttt{BP-vRAN}. Using standard kernel-based methods (as \texttt{BP-vRAN} does) is widely recognized to result in a high computational cost of $\mathcal{O}(t^3)$ with respect to the number of data points $t$ \cite{vakili22}. This is a significant limitation as it delays the vBS operation to more than \SI{10}{\second} after $t=1k$ when tested on an Apple M1 chip with 8-core CPU@\SI{3.2}{\giga\hertz}. Clearly, this hinders the vBS operation, which will then have to rely on stale information. On the other hand, we notice that \texttt{BSvBS} requires no more than \SI{0.08}{\milli\second} to decide a policy, which remains constant throughout.

\begin{figure}[!t]
    \centering
    \label{fig:V2-inference time}\includegraphics[width=0.83\linewidth]{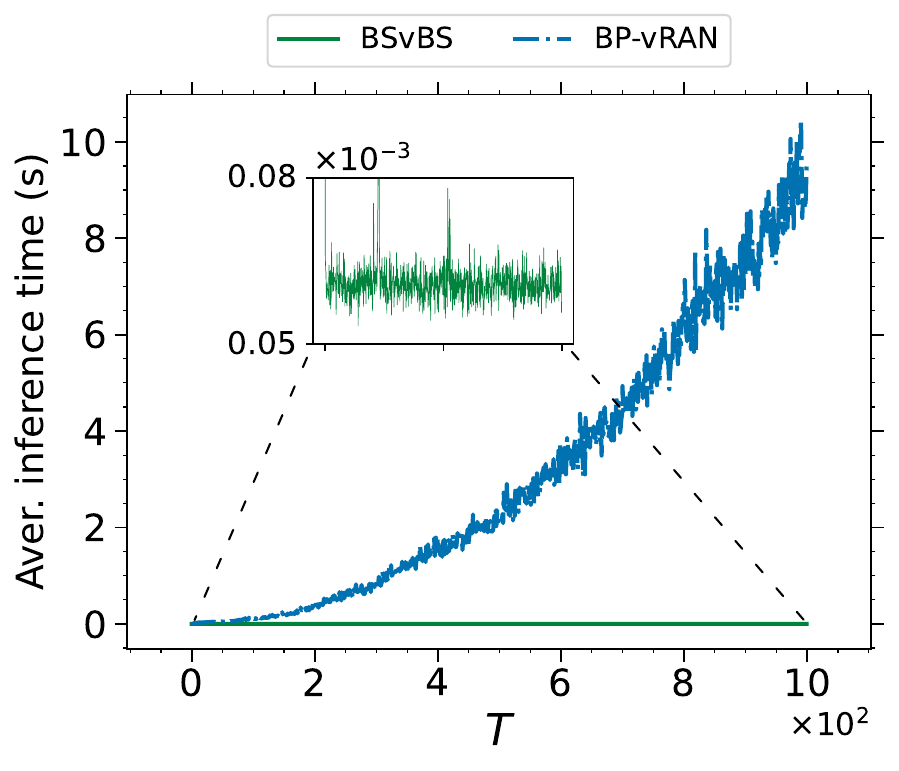} 
    \caption{Average time needed to infer a policy in each round, for our algorithm \texttt{BSvBS}, and its main competitor, \texttt{BP-vRAN}.}
    \label{fig:inference_time_BSvBS}
\end{figure}

\begin{figure*}[!t]
    \centering
    \subfigure[]{\label{fig:V2_LIL_dynamic_regret}\includegraphics[scale=0.29]{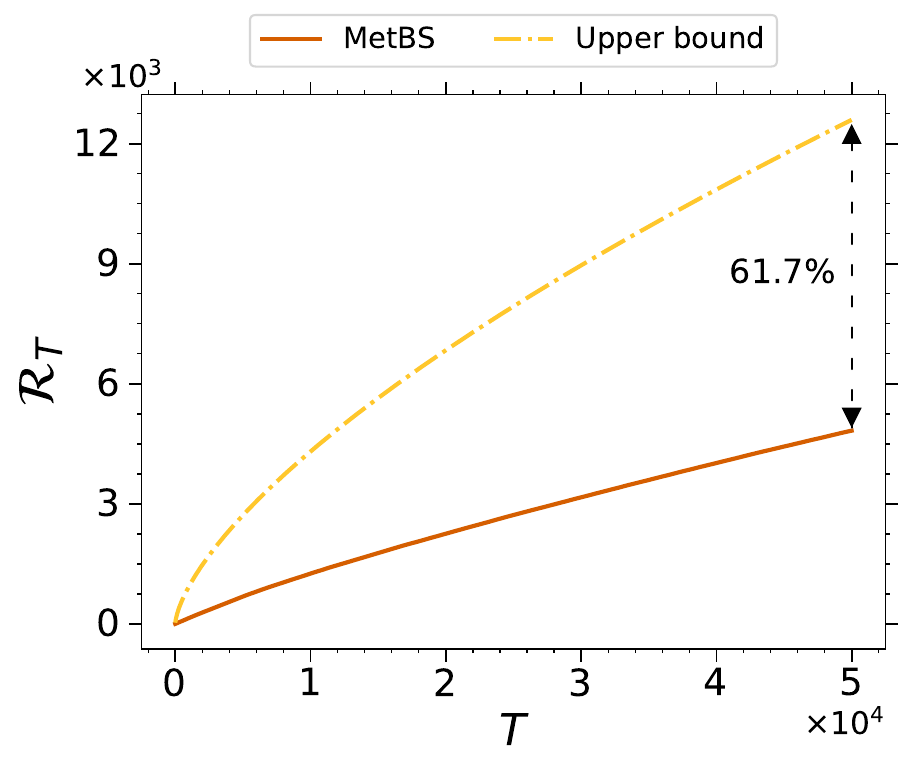}} 
    \subfigure[]{\label{fig:V2_LIL_dynamic_choices}\includegraphics[scale=0.29]{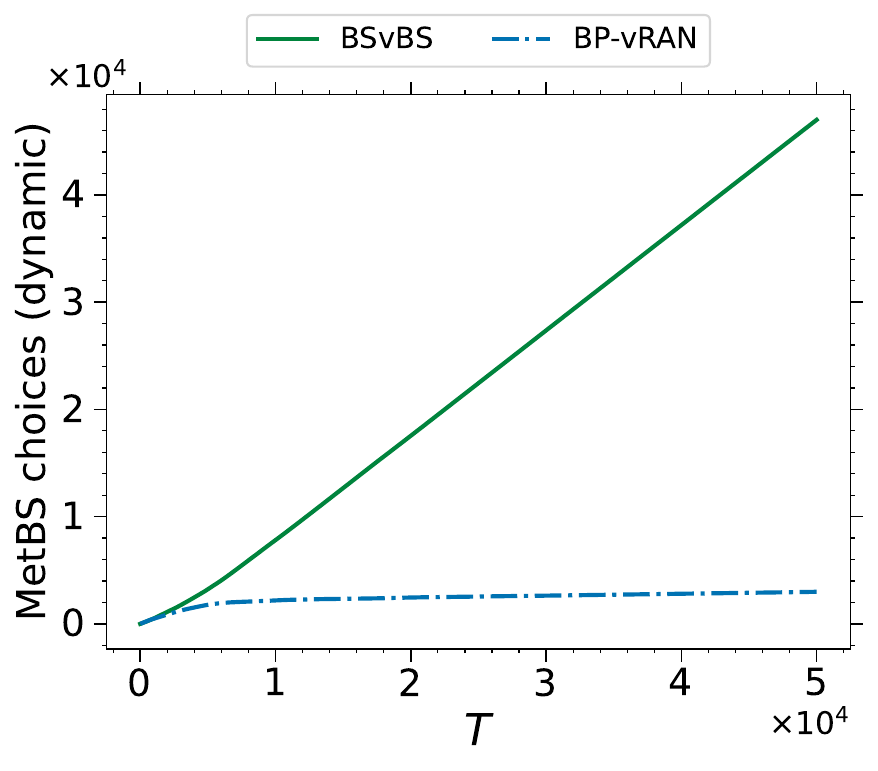}}
    \subfigure[]{\label{fig:V2_LIL_stationary_regret}\includegraphics[scale=0.29]{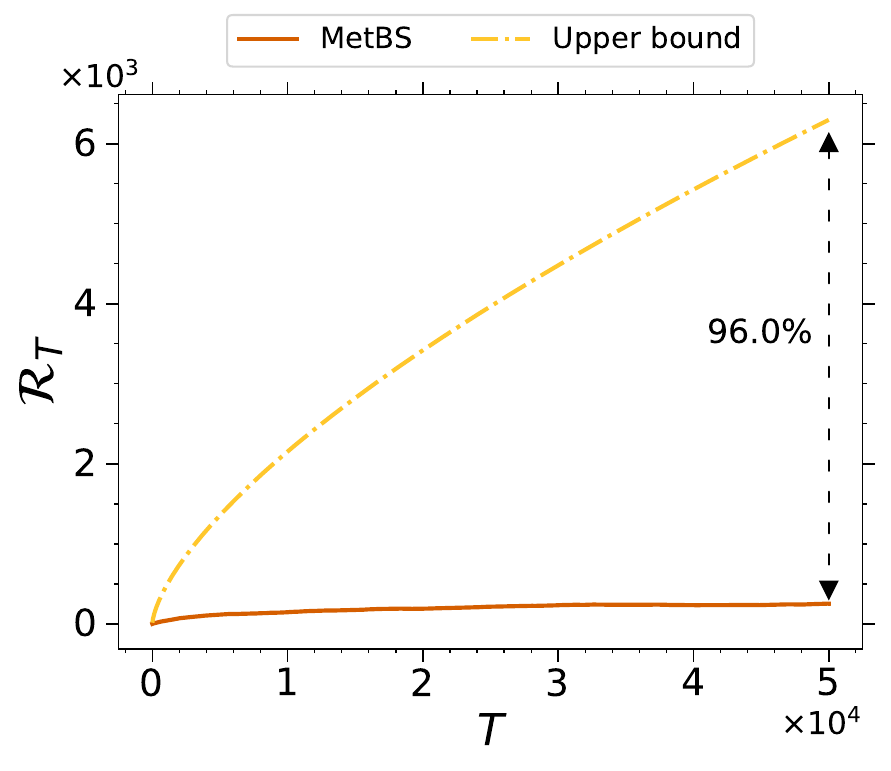}} 
    \subfigure[]{\label{fig:V2_LIL_stationary_choices}\includegraphics[scale=0.29]{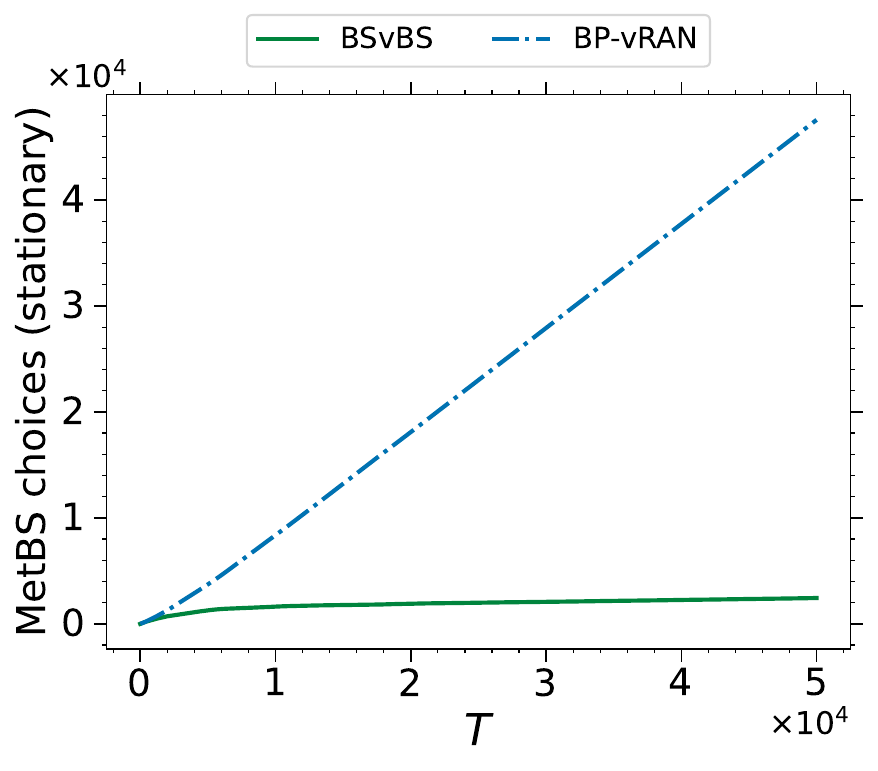}}
    \caption{\textcolor{black}{Meta-learning} algorithm: $R_T$ and the upper bound for dynamic \textbf{(a)} and stationary \textbf{(c)} scenarios; number of times \texttt{BSvBS} and \texttt{BP-vRAN} are chosen in $T=50k$ rounds for dynamic \textbf{(b)} and stationary \textbf{(d)} scenarios.}
    \label{fig:LIL}
\end{figure*}

\noindent\textit{\textbf{{Key takeaways}}}: In  \textcolor{black}{challenging} (i.e., non-stationary / adversarial) environments, decisions for configuring the vBS should be taken based on past performance. Requiring \emph{perfect} knowledge of the environment could lead to sub-optimal policies, increasing power costs up to \SI{64.5}{\percent} for operators. \texttt{BSvBS}'s performance is robust to such adversarial scenarios and outperforms a state-of-the-art algorithm in terms of: \textit{(i)} the average  regret (up to \SI{45.1}{\percent} superiority),  \textit{(ii)} the power gap w.r.t. the minimum vBS energy consumption (up to \SI{64.5}{\percent} superiority), and  \textit{(iii)} inference time (solely \SI{0.08}{\milli\second}). \textcolor{black}{We recall that \texttt{BSvBS} does not have access to how and when the demands and CQI change.}

\subsection{Evaluation of the \textcolor{black}{Meta-Learning Algorithm}}

We consider $A\! = \!2$ with \texttt{BP-vRAN}, and \texttt{BSvBS} that select \textcolor{black}{policies} from $\mathcal{X}$. On the one hand, if the context is not available at the beginning of each round, as happens in several real-world applications, \texttt{BSvBS} is superior and \texttt{BP-vRAN} fails, as seen in Sec. \ref{sec:evaluation}. Hence, \texttt{MetBS} opts mainly for \texttt{BSvBS}. The attained regret (Fig. \ref{fig:V2_LIL_dynamic_regret}) is by \SI{61.7}{\percent} less than the upper bound, which implies the desired sub-linear regret. The algorithms that \texttt{MetBS} chooses can be verified in Fig. \ref{fig:V2_LIL_dynamic_choices}, where \texttt{BSvBS} is selected in approximately $47k$ rounds, while the sub-optimal \texttt{BP-vRAN} in the remaining $3k$ rounds ($T=50k$). On the other hand, if the environment is \textcolor{black}{{easy}}, \texttt{BP-vRAN} is expected to converge faster than \texttt{BSvBS}; and, as a consequence, to be preferred by the meta-learner. Indeed, the regret of \texttt{MetBS} (Fig. \ref{fig:V2_LIL_stationary_regret}) is \SI{96}{\percent} lower than the upper bound stated in \eqref{eq:exp3_bound}, which clearly indicates the expected sub-linear regret has been achieved. \texttt{MetBS} selects \texttt{BP-vRAN} in roughly $46k$ rounds, while \texttt{BSvBS} in $4k$ rounds (Fig. \ref{fig:V2_LIL_stationary_choices}). It is important to heed that \texttt{BSvBS} converges as well to the optimal \textcolor{black}{policy} but slower (see Fig. \ref{fig:static_BSvBS} and Fig. \ref{fig:stationary_BSvBS}), an unavoidable side-effect of its robustness under any environment (even adversarial).

Finally, we test the meta-learner in a ``mixed'' environment, where, in the first $5k$ rounds the demands and CQIs are drawn from Scenario B (stationary), and in the remaining $45k$ rounds from Scenario C (adversarial). Fig. \ref{fig:V2_LIL_change_util} depicts the average rewards of \texttt{MetBS}, \texttt{BSvBS}, and \texttt{BP-vRAN}. It can be viewed that before the change of the environment, the average reward of the meta-learner follows closer to the reward of \texttt{BP-vRAN}; the orange dotted line is \SI{3.8}{\percent} lower than the blue dash-dotted line. The same can be verified from Fig. \ref{fig:V2_LIL_change_probs}, where \texttt{BP-vRAN} is chosen with higher probability, \SI{58}{\percent}, before  $t=5k$. When the change occurs, \texttt{MetBS} does not opt immediately for \texttt{BSvBS}, as the average reward of \texttt{BP-vRAN} is still higher, until the change-point at roughly $t=8k$, which is shown with a red dot in Fig. \ref{fig:V2_LIL_change_util}. After this round, \texttt{BSvBS} experiences larger reward values \textcolor{black}{on average}, and within less than $1k$ rounds (i.e., \SI{2}{\percent}), \texttt{MetBS}  starts indeed selecting \texttt{BSvBS} more frequently (up to \SI{88.2}{\percent}).

\begin{figure}[!t]
    \centering
    \subfigure[]{\label{fig:V2_LIL_change_util}\includegraphics[width=0.83\linewidth]{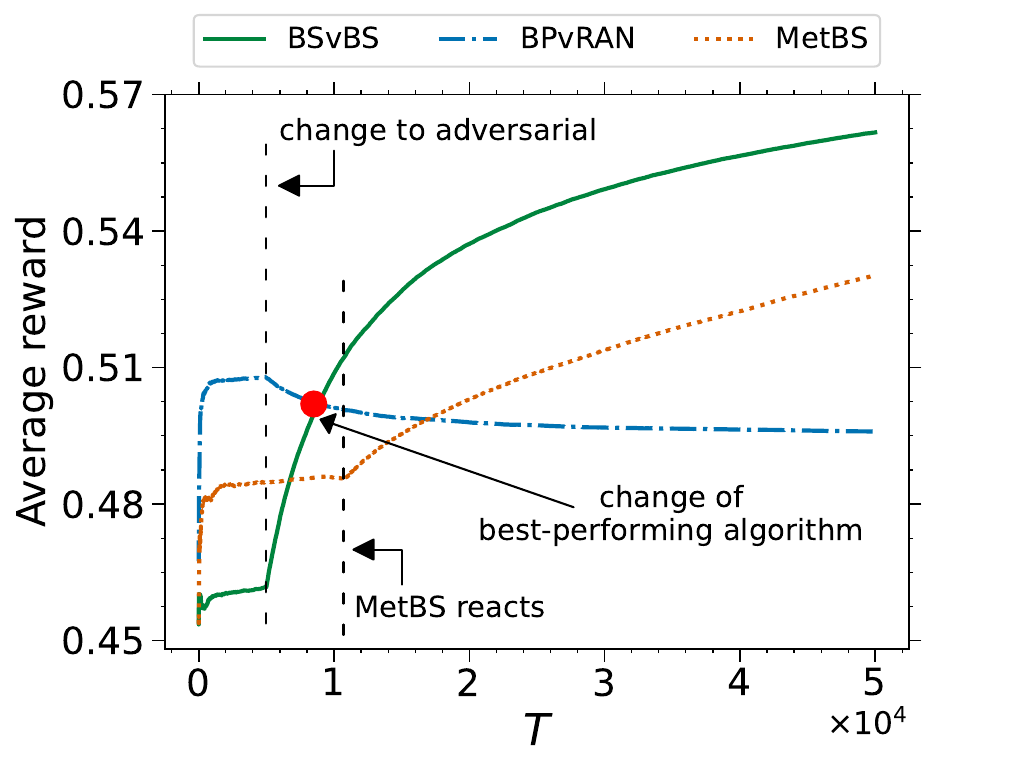}} 
    \subfigure[]{\label{fig:V2_LIL_change_probs}\includegraphics[width=0.76\linewidth]{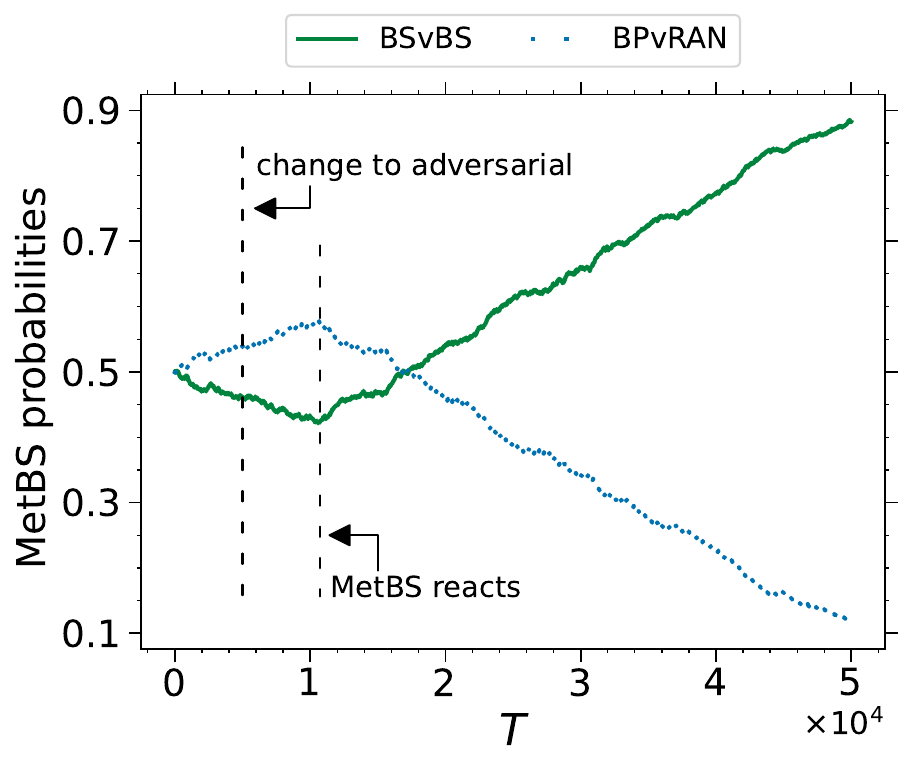}}
    \caption{\textbf{(a)} Average \textcolor{black}{reward} and \textbf{(b)} probabilities that \texttt{BSvBS} and \texttt{BP-vRAN} are chosen by the \texttt{MetBS} when the environment changes (leftmost dashed line, \textit{change \textcolor{black}{to adversarial}}) from stationary to adversarial at $t=5k$. The red dot (\textit{\textcolor{black}{change of best-performing algorithm}}) shows the change of the best-performing algorithm (from \texttt{BP-vRAN} to \texttt{BSvBS}), and the rightmost dashed line \texttt{MetBS} reacts) depicts the round after which \texttt{MetBS} \textcolor{black}{starts choosing} the best-performing \textcolor{black}{algorithm}, \texttt{BSvBS}, more often.}
    \label{fig:LIL_change}
\end{figure}

\noindent\textbf{\textit{{Key takeaways}}}: \texttt{MetBS}  chooses the best-performing algorithm for each scenario. When the demands and CQIs are drawn from a stationary distribution, it prioritizes \texttt{BP-vRAN} (\SI{92}{\percent} of rounds), while in adversarial scenarios, it follows \texttt{BSvBS} (\SI{94}{\percent} of rounds). In mixed scenarios, \texttt{MetBS} tracks \textcolor{black}{and applies} the changes after only \SI{2}{\percent} of rounds.

\section{Conclusions and Future Work} \label{sec:conclusions}

\noindent The virtualization of base stations and the design of O-RAN systems are instrumental for the success of the next generation of mobile networks. \textcolor{black}{Allocating resources for these vBSs by choosing policies} that operate on a longer time scale and do not require intervention on the (often proprietary) vBS node implementations is a new and promising network control approach. However, in order to be practical and successful, \textcolor{black}{the proposed solutions} should have low overhead and \textcolor{black}{require no assumption about} the future \textcolor{black}{channel qualities} and \textcolor{black}{traffic} demands (i.e., the environment). 

\textcolor{black}{The first proposed scheme} possesses exactly these properties, building on a tailored adversarial learning algorithm that has minimal overhead and can run in sub-milliseconds. In line with prior works, we focus on the important metrics of throughput and energy consumption and explore their trade-offs in a range of scenarios with experimental datasets. As this robustness comes at a cost for convergence speed, we aim to increase the latter in \textcolor{black}{{easy}} scenarios, where the \textcolor{black}{environment is} known beforehand (or changes slowly), through a meta-learning scheme that combines a mix of algorithms, including our own, and delineates the best-performing one at runtime. This creates a best-of-both-worlds solution. Our extensive data-driven experiments demonstrate energy savings up to \SI{64.5}{\percent} compared to state-of-the-art competitors. 

We highlight that the regret depends on the number of possible policies up to a square-root factor. While their number is expected to be smaller than the number of policies applied to the RT O-RAN level, this finding still points to an interesting direction for further reducing this dependency. \revminor{Finally, exploring the interactions between rApps and xApps in a real-time setting and assessing their impact on network performance is an interesting direction for future work, by expanding, e.g., \cite{openrangym}.}

\section*{Acknowledgments} \label{sec:acknowledgments}

This work was supported by the European Commission through Grant No. 101139270 (ORIGAMI).

\appendices
\ifCLASSOPTIONcaptionsoff
  \newpage
\fi
\bibliography{references.bib}
\bibliographystyle{IEEEtran}

\vspace{-5mm}
\begin{IEEEbiography}[{\includegraphics
[width=1in,height=1.25in,clip,
keepaspectratio]{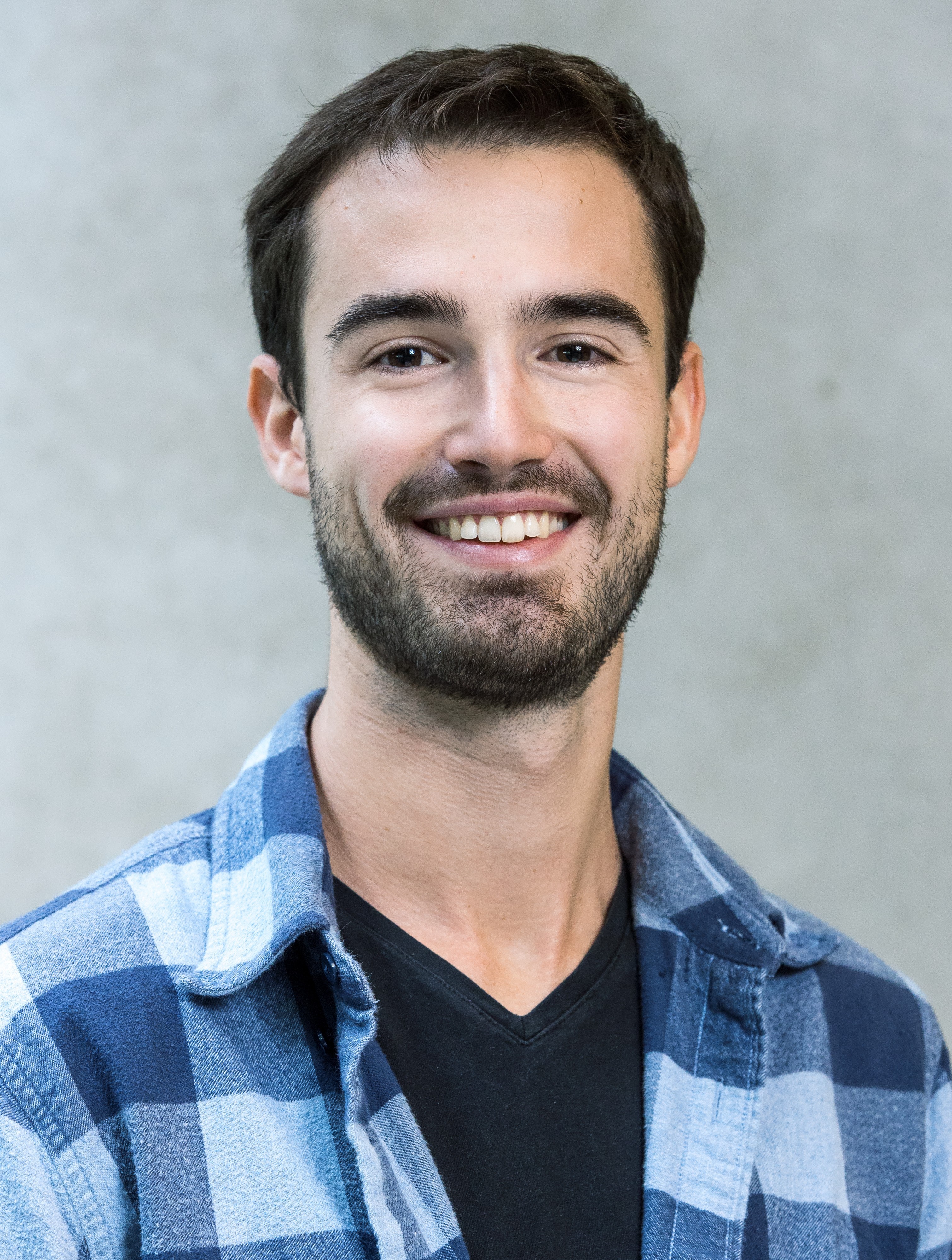}}]
{Michail Kalntis}(Graduate Student Member, IEEE) received the Diploma (5-year joint) degree with the highest possible distinction in Electrical and Computer Engineering (major in Computer Science), from the National Technical University of Athens, in 2021. During his studies, he was a recipient of the Merit Scholarship from Propondis Foundation. He is currently pursuing the Ph.D. degree in Artificial Intelligence \& Machine Learning for Optimization in Wireless Networks with the Delft University of Technology, the Netherlands. His research interests lie in Machine Learning, Network Optimization, and Mobile Networks.
\end{IEEEbiography}

\begin{IEEEbiography}[\vspace{-5mm}{\includegraphics
[width=1in,height=1.25in,clip,
keepaspectratio]{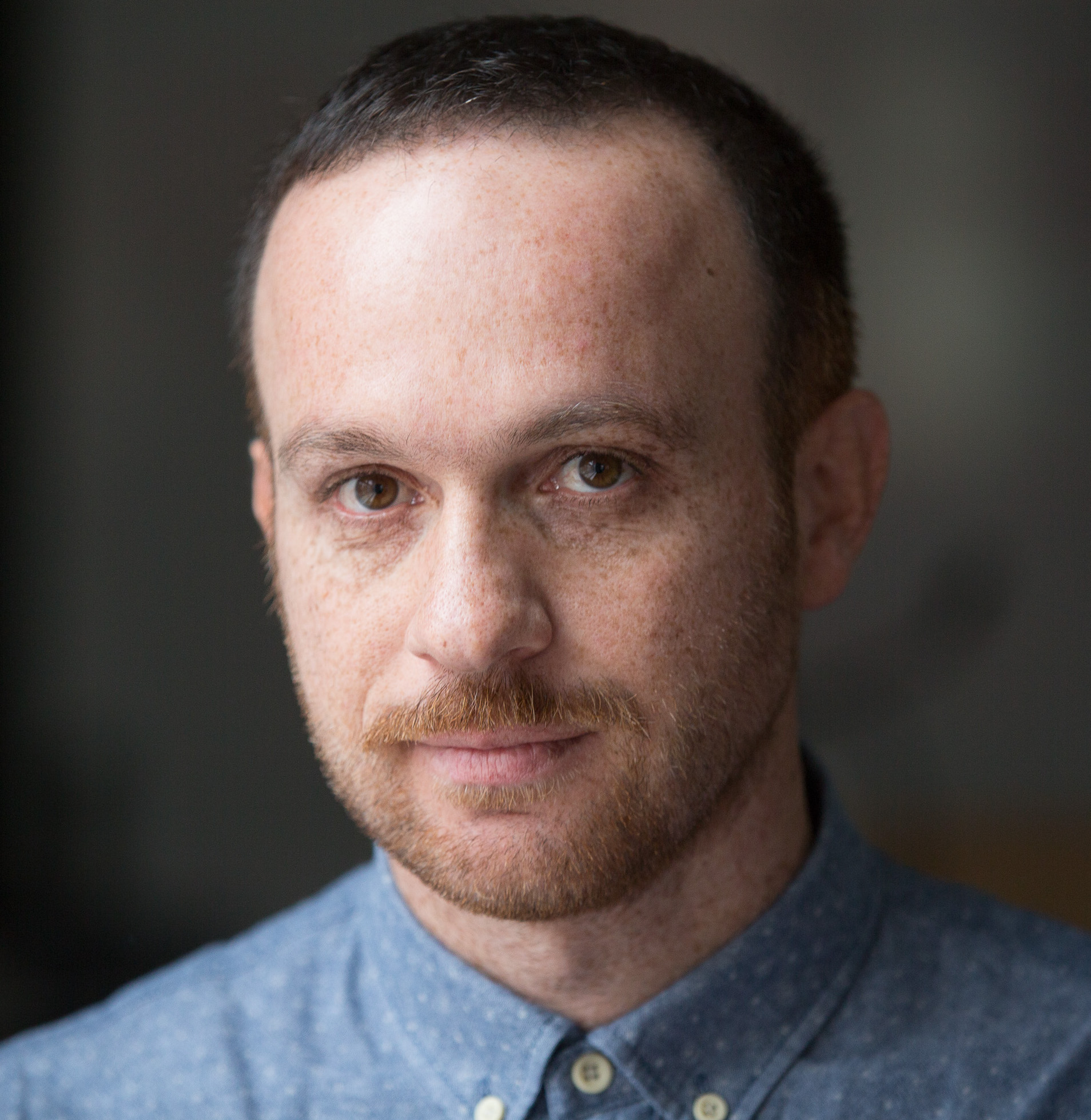}}]
{George Iosifidis}(Member, IEEE) received the Diploma degree in electronics and telecommunications engineering from the Greek Air Force Academy, Athens, in 2000, and the Ph.D. degree from the University of Thessaly in 2012. He was an Assistant Professor with Trinity College Dublin from 2016 to 2020 and a Research Scientist with Yale University (2015-2017). He is currently an Associate Professor with the Delft University of Technology. His research interests lie in the broad area of network optimization and economics; more information can be found at \url{ https://www.futurenetworkslab.net/}.
\end{IEEEbiography}

\begin{IEEEbiography}[{\includegraphics
[width=1in,height=1.25in,clip,
keepaspectratio]{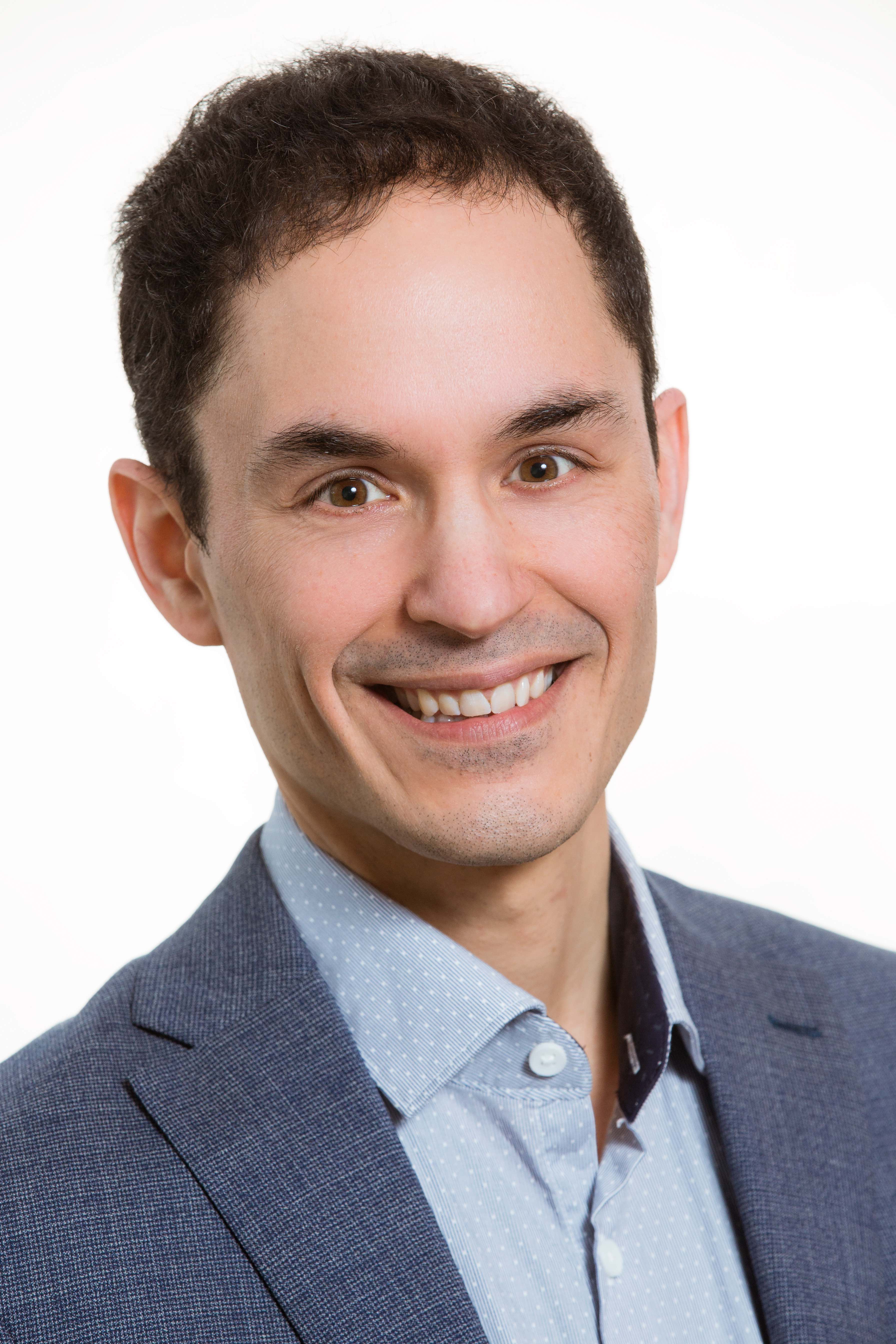}}]
{Fernando A. Kuipers}(Senior Member, IEEE) is a Full Professor at Delft University of Technology (TU Delft), where he established and leads the Networked Systems group and the Lab on Internet Science. He was a Visiting Scholar at Technion, Israel Institute of Technology, in 2009, and Columbia University, New York City, in 2016. His research comprises network optimization, network resilience, quality-of-service, and quality-of-experience and addresses problems in computer networks, software-defined networking, 6G, and Internet-of-Things. His work on these subjects led to a PhD degree conferred Cum Laude in 2004, the highest possible distinction at TU Delft, and includes distinguished papers at IEEE INFOCOM 2003, Chinacom 2006, IFIP Networking 2008, IEEE FMN 2008, IEEE ISM 2008, ITC 2009, IEEE JISIC 2014, NetGames 2015, and EuroGP 2017. He has served as General Chair and TPC Chair in flagship conferences such as ACM SIGCOMM (2021 and 2022) and IEEE INFOCOM (2024), and is Vice Chair of the ACM SIGCOMM Executive Committee. He co-founded the Do IoT fieldlab and the PowerWeb Institute and served on the board of the TU Delft Safety \& Security Institute. Currently, he is co-PI of the Dutch 6G flagship project Future Network Services, where he leads the program line Intelligent Networks.
\end{IEEEbiography}

\end{document}